\documentclass[preprint]{elsarticle}

\usepackage{amsmath}
\usepackage{amssymb}
\usepackage{amsthm}

\usepackage{subfigure}
\usepackage{stmaryrd}
\usepackage{enum}

\usepackage{soul,color}

\usepackage{url}
\pdfmapfile{+sansmathaccent}

\usepackage{xspace}
\usepackage{tikz}

\usetikzlibrary{automata, positioning, arrows}

\tikzset{every picture/.style={node distance=4em,->,>=stealth',shorten >=1pt,auto, font=\scriptsize}} 
\tikzset{initial text=, initial distance=1em} 
\tikzset{accepting by double} 
\tikzstyle{every state}=[minimum size=1em, inner sep=0pt, font=\scriptsize]

\newcommand*{\gob}{\textrm{\tiny $\bullet$}}
\newcommand*{\gout}{\raisebox{.5pt}{\kern-1pt\textrm{\tiny $\times$}\kern-1pt}}
\newcommand*{\goc}{\raisebox{-1.2pt}{\textrm{\tiny
      $\bullet$}}\kern-3.7pt\raisebox{2pt}{\textrm{\tiny $\bullet$}}}

\newcommand{\takeout}[1]{}
\newcommand{\calA}{\mathcal{A}}
\newcommand{\calG}{\mathcal{G}}
\newcommand{\calK}{\mathcal{K}}
\newcommand{\calD}{\mathcal{D}}
\newcommand{\calP}{\mathcal{P}}

\newcommand{\bbZ}{\mathbb{Z}}
\newcommand{\bbN}{\mathbb{N}}

\newcommand{\cbot}{\mathit{bot}}

\newcommand{\CKS}{\textsc{cks}\xspace}%
\newcommand{\RCKS}{$\omega$-\textsc{cks}\xspace}%
\newcommand{\RMKS}{$\omega$-\textsc{mks}\xspace}%
\newcommand{\MKS}{\textsc{mks}\xspace}%

\newcommand{\CK}{\mathit{CK}}
\newcommand{\State}{\mathit{Ter}}
\newcommand{\outcome}{\mathrm{outcome}}

\newcommand*{\cut}[1]{\kern-0.1em\restriction_{#1}\kern0.0em}
\newcommand*{\tck}[1]{{#1}^{\tiny \textsc{ck}}}

\newcommand*{\igap}[2]{\llbracket #1,#2 \rrbracket}

\usepackage{pifont}

\newcommand{\pass}{\mathit{in}}
\newcommand{\call}{\mathit{out}}

\newcommand{\PSPACE}{\textsc{PSpace}}
\newcommand{\NLOGSPACE}{\textsc{NLogSpace}}

\newcommand{\NEXPTIME}{\textsc{NExpTime}\xspace}
\newcommand{\NTIME}{\textsc{NTime}}
\newcommand{\EXPTIME}{\textsc{ExpTime}}

\newtheorem{theorem}{Theorem}[section]

\newtheorem{prop}[theorem]{Proposition}
\newtheorem{lemma}[theorem]{Lemma}

\newtheorem{corol}[theorem]{Corollary}

\theoremstyle{definition}

\begin{document}

\begin{frontmatter}

\title{Infinite games with finite knowledge gaps}

\author[LSV]{Dietmar Berwanger}
\ead{dwb@lsv.fr}

\author[LSV,IMSC]{Anup Basil Mathew}
\ead{anupbasil@imsc.res.in}

\address[LSV]{LSV, CNRS \& ENS Cachan, Universit\'e Paris-Saclay, France}
\address[IMSC]{Institute of Mathematical Sciences, Chennai, India}

\begin{abstract}
Infinite games where several players seek to coordinate under
imperfect information are deemed to be undecidable, unless the
information is hierarchically ordered among the players.

We identify a class of games for which joint winning
strategies can be constructed effectively 
without restricting the direction of information flow. 
Instead, our condition
requires that the players attain common knowledge 
about the actual state of the game 
over and over again along every play.
 
We show that it is decidable whether a given game satisfies the
condition, and prove tight complexity bounds for the  
strategy synthesis problem under $\omega$-regular winning conditions given
by deterministic parity automata.
\end{abstract}


\end{frontmatter}


\section{Introduction}

Automated synthesis of systems that are correct by construction 
is a persistent ambition of computational engineering. 
One major challenge consists in controlling 
components that have only partial information about the 
global system state. 
Building on automata and game-theoretic foundations,
significant progress has been made
towards synthesising finite-state components that interact 
with an uncontrollable environment either individually, 
or in coordination with other controllable components\,---\,%
provided the information they 
have about the global system is distributed hierarchically~\cite{PnueliRos89,KupfermanVar01}.
For the general case, however, it was shown that 
the problem of coordinating 
two or more components of a distributed system with 
non-terminating executions is
undecidable~\cite{PnueliRos90,ArnoldWal07}. 

The distributed synthesis problem can be formulated alternatively
in terms of games between
$n$~players (the components) that move along the edges of a finite
graph  
(the state-transitions of the global system) with imperfect information about the current position and the moves of the other players. 
The outcome of a play 
is a possibly infinite path (system execution) 
determined by the joint actions of the players 
and moves of Nature (the uncontrollable environment). 
The players have a common winning condition:
to form a path that corresponds to a correct 
execution with respect to 
the system specification, no matter how Nature moves.
Winning conditions may be specified by finite-state automata, temporal logics, 
or in the canonical form of parity conditions.
Thus, distributed synthesis under partial information corresponds to 
the 
problem  of constructing a winning profile of finite-state 
strategies 
in a coordination game with imperfect information. 
This problem was shown to be undecidable by Peterson and 
Reif~\cite{PetersonRei79}, 
already for the basic setting of two players 
with a reachability condition; 
infinitary winning conditions, which lead to higher degrees of
undecidability, have been studied by Janin~\cite{Janin07}.

As in the case of distributed systems, 
decidable classes of coordination games 
rely on restrictions of the
information flow according to an order among the 
players~\cite{AzharPetRei01,Kai06,BerwangerMatVDB15}.
In their survey article on the complexity of multiplayer
games, Azhar, Peterson, and Reif conclude that
``[i]n general, multiplayer games of incomplete information can be
undecidable, unless the information is 
hierarchically arranged''\cite[p. 991]{AzharPetRei01}. 

The undecidability arguments cited above share a basic scenario: 
two players become uncertain 
about the current state of the game, due to moves of Nature.
The structure of the game requires them to 
take into account not only their 
first-order uncertainty about the actual state, 
but also the higher-order uncertainty of 
one player about the knowledge of the other.
Finally, the players can win only by attaining
common knowledge about a property of the actual 
history, which requires them to 
maintain knowledge hierarchies of 
increasing height, as the play proceeds. 
The scenario is set up so that the 
uncertainty never vanishes and the knowledge
hierarchies grow unboundedly, which leads to 
undecidability~\cite{BerwangerKai10}.
  
One systematic approach to characterising
undecidable classes of problems that involve 
multiple players with imperfect information is the
 \emph{information fork} criterion formulated 
by Finkbeiner and 
Schewe~\cite{FinkbeinerSch05}. 
The criterion applies to the distributed synthesis problem 
in the basic setting of Pnueli and Rosner~\cite{PnueliRos90},
for fixed architectures with
synchronous communication channels. 
Intuitively, 
an architecture has an information fork if it allows for 
two players to reach a situation 
in which neither one can infer
the observation received by the other player 
from his own observation.
Under this condition,  
distributed architectures may allow
the knowledge of players to diverge
over an unbounded number of rounds for particular specifications.
Conversely, all architectures 
that do not contain an information fork 
admit an information ordering among the players and are therefore
decidable.

When applied to games, however,
the information fork criterion yields only a 
coarse classification for decidability, 
as the parametrisation over architectures 
has no natural correspondent in terms of 
game graphs.
Indeed, the set of game
instances obtained by modelling a
given family of distributed systems 
may be solvable, in spite of possible
information forks in the underlying architecture. 
This can occur, for instance, if the information flow
affected by the fork is inessential to the players, 
or if the divergence between the knowledge of players 
vanishes after few rounds.

In this article, we identify a new condition for the decidability of 
coordination games with imperfect information that does not rely 
on the hierarchical arrangement of information. 
Similar to the information fork approach, our focus is on
situations in which the knowledge of players
diverges. We use the term
\emph{knowledge gap} to describe an interval of rounds 
at which the players do not attain common knowledge 
about the actual state. Essentially, our condition 
requests that all knowledge gaps of a game are finite, 
or in other words, that the players attain common knowledge 
of the actual game state
infinitely often, along every play.   
In this case we say that the game allows 
for \emph{recurring common knowledge of the state}.

Questions about common knowledge 
in infinite runs are typically hard.  
In their study of 
the model checking complexity of epistemic temporal
logics~\cite{vdMeydenShi99}, 
van der Meyden and Shilov 
point out that already the problem of determining 
whether the players attain common knowledge about 
an atomic property is
undecidable, even for synchronous models as we consider here. 
Indeed, it turns out that
it is undecidable whether 
the players attain
common knowledge of the state at any history within a 
given part of the game graph (Proposition~\ref{prop:ck-existential}). 

Surprisingly, the situation improves 
when we look at the recurring formulation relevant 
for our characterisation: We are able to show 
that the question of whether the common-knowledge property holds
infinitely often, on every play in a game is decidable with low
complexity. This has several favourable consequences
for solving infinite coordination games with imperfect information.

Our results are summarised as follows:
\begin{enumerate}[(1)]
\item The question of whether a game for $n$ players 
  with imperfect information satisfies the condition 
  of recurring common knowledge of the state is decidable in $\NLOGSPACE$.
\item If a coordination game for $n$ players 
  with imperfect information satisfies the condition of recurring
  common knowledge of the state, then the problem of whether a joint 
  winning strategy exists is decidable, and it is 
  $\NEXPTIME$-complete.
\item If there exists a joint winning strategy in a game
  with recurring common knowledge of the state, then there also 
  exists a profile of finite-state strategies 
  of exponential size, which can be synthesised in $2\EXPTIME$.
\end{enumerate}

The conclusions rely on three key arguments. Firstly, we show that
under recurring common knowledge of the state,
the intervals where the current state of the game is 
not common knowledge are 
bounded uniformly. This implies that 
the \textit{perfect-information tracking} of
such a game is finite, which yields decidability of 
the strategy synthesis problem as a consequence 
of a metatheorem from~\cite{BKP11}. Secondly, we characterise 
recurring \emph{common} knowledge in terms of 
recurring \emph{mutual} knowledge.  
This allows us to establish tight complexity bounds.
Finally, we prove that the problem of solving 
imperfect-information games with recurring common knowledge of the
state can be 
reduced to solving parity games with perfect information, at a
relatively low cost in terms of complexity.

\section{Basic notions}

\subsection{Coordination games with imperfect information}

Our game model is close to that of concurrent games \cite{AlurHK02}.
There are $n$ players $1$, \dots, $n$ and a distinguished agent 
called Nature. The \emph{grand coalition} is the set 
$\{1, \dots, n \}$  
of all players. 
We refer to a list of elements
$x=(x^i)_{1\le i \le n}$, one for each player, as a \emph{profile}.  

For each player~$i$, we fix a set $A^i$ of \emph{actions} and a set $B^i$ of \emph{observations}, finite unless stated otherwise.
The \emph{action space} $A$ consists of all action profiles.
A~\emph{game graph} $G = (V, E, (\beta^i)_{1 \le i \le n})$ 
consists of a finite set $V$ of nodes called \emph{states}, 
an edge relation $E \subseteq V \times A \times V$ 
representing simultaneous \emph{moves} 
labelled by action profiles, 
and a profile of \emph{observation} functions $\beta^i: V \to B^i$
that label every state with an observation, for each player.
We assume that for every
state $v$ and every
action profile~$a$ there is an outgoing move $(v, a, v') \in E$. 
For convenience, we will 
include special sink states from which any
outgoing move is a self loop. 

Plays start at an initial state $v_0 \in V$ known to all players and proceed in rounds where each player~$i$ chooses an action $a^i \in A^i$, then  
Nature chooses a successor state $v'$ 
reachable via a move $(v, a, v') \in E$, and
each player~$i$ receives the observation $\beta^i( v')$. 
Notice that the players are not informed about the action
chosen by other players, nor about the state chosen by Nature. 

Formally, a \emph{play} 
is an infinite sequence $\pi = v_0, a_1, v_1, a_2, v_2, \dots$ 
alternating between positions and action profiles 
with $(v_{\ell}, a_{\ell+1}, v_{\ell + 1}) \in E$, 
for all $\ell \ge 0$.
A~\emph{history} is a prefix $\pi = v_0, a_1, v_1, \dots, a_{\ell},
v_{\ell}$ of a play; we refer to $\ell$ as the
\emph{length} of the history.
For convenience, we omit commas in the sequence notation and 
write plays and histories as words 
$\pi = v_0\,a_1 v_1\, a_2 v_2 \dots$. 
Whenever we refer to a
finite prefix $\rho$ of a play or history $\pi$, 
we mean a history $\rho \in V(AV)^*$ such that $\pi = \rho\tau$ for
some $\tau$ in $(AV)^*$ or $(AV)^\omega$; 
further, we call $\pi$ a
\emph{prolongation} and~$\tau$ a
\emph{continuation} of $\rho$.

The observation function is extended from states 
to histories and plays 
by setting $\beta^i( \pi ) := \beta^i( v_0 )\, a_1^i \beta^i (v_1) \dots$
We say that 
two histories $\pi, \pi'$ are \emph{indistinguishable} to Player~$i$, and write 
$\pi \sim^i \pi'$, if $\beta^i(\pi) = \beta^i(\pi')$. 
This is an equivalence relation, and its classes are 
called the \emph{information sets} of Player~$i$. 
A game (graph) 
with \emph{perfect information} is one where all information sets are
singletons. In general, we do not assume that this is the case, so
we speak about games with \emph{imperfect information}.

When viewed as a distributed system in the taxonomy of
Halpern and Vardi~\cite{HalpernVar89}, 
our game model belongs to the
class of \emph{synchronous} systems 
with \emph{perfect recall}. 
This is implicit in our definition of observation functions:
the players are able to distinguish between 
histories of different length (synchronicity), 
and if two histories are indistinguishable 
for a player~$i$ at round~$\ell$, 
then so are they at any previous round~$r < \ell$ 
(perfect recall).

A \emph{strategy} for Player~$i$ is a mapping $s^i: V(AV)^* \to A^i$
from histories to actions
such that  $s^i ( \pi ) = s^i( \pi')$, 
for any pair ~$\pi \sim^i \pi'$ 
of indistinguishable histories.
We denote the set of all strategies of Player~$i$ 
by $S^i$ and the set of all strategy profiles by $S$.
A~history or play $\pi = v_0 \, a_1  v_1 \dots$ 
\emph{follows} the strategy $s^i \in S^i$ if  
$a_{\ell+1}^i = s^i( v_0 \, a_1 v_1 \dots a_{\ell} v_{\ell})$, 
for every  $\ell \ge 0$.
For the grand coalition, the play $\pi$ follows a strategy
profile~$s \in S$ if it follows all strategies~$s^i$.
The set of  possible \emph{outcomes} of a strategy profile~$s$ 
is the set of plays that follow~$s$.

A general \emph{winning condition} over a game graph $G$ is a set 
$W \subseteq (V\kern-2pt A)^\omega$ of plays. 
A~\emph{coordination game} $\calG = (G, W)$ 
is described by a game graph and a winning condition. 
We say that a play $\pi$ on $G$ is winning in~$\calG$ if $\pi \in
W$. 
A~strategy profile $s$ is winning in~$\calG$, 
if all its possible outcomes are so. In this case, we refer to $s$ as
a \emph{joint winning strategy}. 

With a view to effective algorithms for synthesising strategies, 
we focus on finitely presented games where the winning
condition is described by 
a colouring function $\gamma: V \to C$ with a finite
range of colours, and an $\omega$-regular set $W \subseteq C^\omega$
given, e.g., by finite-state automaton. 
Then, a~play $v_0 \, a_1 v_1 \dots$ is winning if
$\gamma( v_0 )  \gamma( v_1) \dots  \in W$.
We generally assume that the colouring is \emph{observable} 
to each player~$i$, 
that is, $\beta^i( v ) \neq
\beta^i( v')$ whenever $\gamma( v ) \neq \gamma( v')$.
Given such a game, 
the \emph{distributed synthesis problem} consists
of two tasks: (1) to decide whether there exists a joint winning strategy, and
(2) to construct a winning profile of finite-state strategies, if
any exist.  These are 
strategies implemented by automata that input observations and output actions.
For more background on finite-state strategy synthesis
we refer to the expository article of Thomas~\cite{Thomas95}.

For lower bounds, we refer to simple safety conditions 
which require the players to avoid an observable sink~$\ominus$.
The technical results on upper bounds are formulated 
in terms of \emph{parity} winning conditions represented by 
a coloring function 
$\gamma: V \to \mathbb{N}$ that maps every state to a number called 
\emph{priority}: A play is winning if the least
priority seen infinitely often along a play is even.
Parity conditions provide a canonical form 
for observable $\omega$-regular winning conditions, 
in the sense that each game with a regular
condition can be reduced to one with a parity condition 
such that the solution of the synthesis problem is preserved.
The reduction for the perfect-information setting described 
by Thomas in the handbook 
chapter~\cite{Thomas90} 
generalises to imperfect-information games 
with observable winning conditions, as pointed out in~\cite{BKP11}.

\subsection{Domino tiling problems}

As a tool for proving lower complexity bounds, 
we use domino tiling problems, which 
allow a more transparent representation of combinatorial
problems than encoding machine models.
Our exposition follows the notation of
B\"orger, Gr\"adel, and Gurevich 
\cite{BoergerGraGur97}.

A \emph{domino system} $\calD = (D, E_H, E_V)$ 
is described by a finite set~$D$ of \emph{dominoes} 
together with horizontal and vertical compatibility relations 
$E_H, E_V \subseteq D \times D$.
The generic domino tiling problem 
is to determine, for a given system~$\calD$, 
whether copies of the dominoes can be arranged to
cover a region
$Z \subseteq \bbZ \times \bbZ$, 
such that any two vertically or horizontally 
adjacent dominoes are compatible. 
Here we consider finite rectangular regions 
  $Z( \ell, m ) := \{0, \dots, \ell+1\} \times \{0, \dots, m\}$
where the first and the last column, and the bottom row are  
dis\-tinguished as border areas to be tiled with 
special dominoes~$\#$ and~$\square$.
The concrete question is 
whether there exists a \emph{tiling} 
$\tau : Z( \ell, m) \to D$ 
that assigns to every point 
in the region a domino, 
subject to the border constraints:
\begin{itemize}[-]
\item
  $\tau( x, 0 ) = \square$, for all $x = 1, \dots, \ell$, and
\item
  $\tau( 0, y ) = \tau( \ell + 1, y ) = \#$, for all $y = 0, \dots, m$,
\end{itemize}
and the compatibility constraints, for all~$x \le \ell$ and $y < m$: 
\begin{itemize}[-]
\item if $\tau(x, y) = d$ and $\tau(x + 1, y) = d'$
then $(d, d') \in E_H$, and
\item if $\tau(x, y) = d$ and $\tau(x, y + 1) = d'$
then $(d, d') \in E_V$.
\end{itemize}
In addition, we may specify constraints on the
\emph{frontier} of the tiling, that is, the sequence 
$\tau( 1, m), \tau(2, m), \dots, \tau( \ell, m)$ 
of dominoes in the top row. 
To ensure that border dominoes do not 
appear at the interior of a correct tiling and to avoid the trivial
tiling, we generally assume that 
$E_V \subseteq D \times (D \setminus \{\#, \square\}) \cup \{ (\#, \#)
\}$ and $(\square, \square) \not\in E_H$.

We use three variants of the domino problem.
Firstly, the \textsc{Corridor Tiling} 
problem takes as input a domino system~$\calD$ together
with a frontier constraint $w \in D^\ell$ and asks
whether there exists a height $m$ such that the region
$Z( \ell, m )$ allows a tiling~$\tau$
that additionally satisfies:
\begin{itemize}[-]
\item 
  $\tau( i, m ) = w_i$, for all $i = 1, \dots, \ell$.
\end{itemize}

Secondly, the \textsc{Corridor Universality} problem takes as input 
a domino system~$\calD$ together with a 
subset of dominoes~$\Sigma \subseteq D$ 
and asks whether for \emph{all} 
frontier constraints $w \in \Sigma^\ell$ of 
arbitrary length~$\ell > 0$, there exists a height~$m$ such that
the region $Z( \ell, m )$ allows
a corridor tiling. 

The basic variant of corridor tiling 
is a well-known~$\PSPACE$-complete problem~\cite{Boas97}. 
One way to explain the complexity of both variants 
is via the correspondence between context-sensitive
languages and
domino systems, established by Latteux and 
Simplot~\cite{LatteuxSim97a,LatteuxSim97b}. 
The \emph{frontier language} of a domino system~$\calD$ 
is the set~$L( \calD )$ of words $w \in D^*$ 
such that $(\calD, w)$ 
yields a positive instance of the \textsc{Corridor Tiling} 
problem. 
We refer to standard notions on
context sensitive languages as found, for instance, in 
the handbook~\cite[Chapter 3]{Salomaa73}.

\begin{theorem}[\cite{LatteuxSim97a,LatteuxSim97b}]\label{thm:cs-domino}
For every context-sensitive language~$L \subseteq \Sigma^*$ 
given as a linear bounded automaton, 
one can
construct in polynomial time a
domino system $\mathcal{D}$  
over a set of dominoes $D \supseteq \Sigma$, such 
that~$L( \calD ) \cap \Sigma^* = L$.
\end{theorem}
Via this correspondence, 
the  membership problem for context-sensitive language, 
which is $\PSPACE$-complete, reduces to \textsc{Corridor Tiling}
and the universality problem for context-sensitive languages, 
which is undecidable, 
reduces to \textsc{Corridor Universality}.
Converse reductions from domino tiling to context-sensitive language 
problems can also be done in polynomial 
time, by translating domino systems into linear-bounded automata. 

\newpage

\begin{theorem}[\cite{Boas97}, \cite{LatteuxSim97a,LatteuxSim97b}]%
  \label{thm:corridor-complexity}\leavevmode
\begin{enumerate}[(i)]
\item \textsc{Corridor Tiling} is $\PSPACE$-complete.
\item \textsc{Corridor Universality} is undecidable.
\end{enumerate}
\end{theorem}

Finally, the \textsc{Exp-Square Tiling} problem 
takes as input a domino
system together with a number~$\ell \in \mathbb{N}$ in binary encoding, 
and asks whether the region
$Z( \ell, \ell )$ allows a correct tiling.
The problem was first studied by F\"urer~\cite{Fuerer84}.

\begin{theorem}[\cite{Fuerer84}]\label{expsquare-complexity} 
\textsc{Exp-Square Tiling} is $\NEXPTIME$-complete.
\end{theorem}

\subsection{Common knowledge}

We use the notion of knowledge in the sense of having information. 
That Player $i$ knows proposition $\varphi$ at history $\pi$ 
should mean that, from the structure of the game graph 
and the sequence 
$\beta^i( \pi )$ of observations she received, 
it can be inferred that $\varphi$ holds. 
Specifically, we are interested in propositions about play histories. 
To formalise knowledge and uncertainty, we rely on the 
standard semantic model introduced by Aumann \cite{Aumann76} and
follow the treatment of Osborne and
Rubinstein~\cite[Chapter~5]{OsborneR94}. 
For an extensive account of distributed knowledge in computational
systems, we refer the reader to the 
book of Fagin, Halpern, Moses, and Vardi
~\cite[Chapters 10, 11]{FaginHMV95} and, as a standard reference on 
common knowledge, to the handbook chapter of 
Geanakoplos~\cite{Geanakoplos94}. 
The enlightening article of~\cite{Barwise88}
addresses foundational issues about the formalisation of
common-knowledge.

Let us fix a game graph $G$ and denote by $\Omega$ 
the set of histories. The \emph{possibility} correspondence 
$P^i : \Omega \to \calP(\Omega)$
associates to each history $\pi$ its information set:
\begin{align*}
  P^i( \pi ) := \{ \pi' \in \Omega~\mid~ \pi' \sim^i \pi \}, \quad 
  \text{for every player~$i$.}
\end{align*}
Thus, at history $\pi$, Player~$i$
knows that the actual history is in $P^i(\pi)$, but he may be
uncertain which one it is.
The sets $P^i(\pi)$ form a partition of $\Omega$. 
Observe that each information set $P^i( \pi )$ 
consists of histories of the same length as $\pi$, hence it is finite.

An \emph{event} is a subset $F \subseteq \Omega$. We say that $F$
\emph{occurs} at history $\pi$ if $\pi \in F$. The
\emph{knowledge} operator $K^i: \calP(\Omega) \to \calP(\Omega)$
associates to every event $F$ the set of histories at which Player~$i$
knows that $F$ occurs: 
\begin{align*}
  K^i( F ) := \{ \pi \in \Omega ~\mid~ P^i( \pi ) \subseteq F \},   
  \quad \text{for every player~$i$.}
\end{align*}
Note that $K^i( F )$ is itself an event. If $\pi \in K^i( F )$, 
then (the occurrence of) $F$ is private knowledge 
of Player~$i$ at $\pi$: 
For any $\pi' \sim^i \pi$, it holds that $\pi' \in F$.
If, moreover, $\pi \in K^i ( F )$, for all players~$i$,  
we say that $F$ is \emph{mutual knowledge}
among the players at $\pi$. 
In this case, every player knows that $F$ occurs,
although one player may be uncertain about 
whether another player knows it.

An event $F \subseteq \Omega$ is \emph{common knowledge} at $\pi$ 
if for every sequence of histories $\pi_1, \dots, \pi_k$ and players 
$i_1, \dots, i_k$ such that 
$\pi \sim^{i_1} \pi_1 \sim^{i_2} \dots \sim^{i_k} \pi_k$,
it is the case that $\pi_k \in F$. In other words, 
$\pi$ belongs to the image of $F$ under every iteration 
$K^{i_1}(K^{i_2}( \dots
K^{i_k} ( F ) \dots ))$ of knowledge operators. 
\takeout{We denote by 
  $\CK( F )$ the set of histories at which the event 
  $F$ is common knowledge. Admitting a lax notation, this can be written as:  
  \begin{align*}
    \CK( F ) = \bigcap \{\, K^{i_1}(K^{i_2}( \dots K^{i_k} ( F ) \dots
    ))~|~ 1 \le i_1, \dots, i_k \le n \,\}. 
  \end{align*}
}

We will use an alternative characterisation in terms of 
shared information. 
An event $F$ is \emph{self-evident} if it is mutual knowledge among
the players at every history in $F$, 
that is, if $F \subseteq K^i( F )$ for all players~$i$. As 
the converse inclusion $K^i( F ) \subseteq F$ always holds, 
this amounts to saying that $F$ is a simultaneous 
fixed point of the player's knowledge operators. 
Self-evident sets allow an interpretation of common knowledge 
that coincides with the iterated-knowledge interpretation, if 
the situation model is sufficiently simple 
(see Barwise~\cite{Barwise88}, for a thorough analysis), 
particularly if the sample space~$\Omega$ is finite.
Although in our setting~$\Omega$ is infinite, 
for histories~$\pi$ of length~$\ell$, only
the finite space~$\Omega_\ell$ of histories of the same length matters: 
An event $F \subseteq \Omega$ is mutual or common knowledge at~$\pi$
if, and only if, 
this holds for the event $F \cap \Omega_\ell$.
Therefore, the argument for the finite setting
given, for instance, in the handbook 
chapter of Geanakoplos~\cite[Section 6]{Geanakoplos94}, 
justifies the following characterisation.

\begin{theorem}[\cite{Aumann76}]\label{thm:ck-evident}
An event $C \subseteq \Omega$ is common knowledge at history $\pi$, 
if and only if, there exists a self-evident event $F \subseteq C$ with
$\pi \in F$.
\end{theorem}

\section{Uncertainty and coordination}
\label{sec:uncertainty}

Under perfect information, coordination games are trivial to solve, 
by considering the two-player zero-sum game where 
the observations and the actions of the grand coalition are attributed 
to the first player and the role of Nature is played by the second player.
Then, any winning strategy of the first player can be viewed
as a joint winning strategy and vice versa. 
Intuitively, players of the coalition can act as one
because each player
knows the actual history when he chooses an action.
Unlike the case considered in~\cite{Berwanger09}, 
where players choose their strategy independently, 
in the setting of distributed strategies
there is no risk of discoordination due to strategic uncertainty: 
The joint strategy is centrally planned, 
it is common knowledge among the players. 

Under imperfect information, 
the problem is more complex,
because players may not know where they are in the game. 
Strategies need to adjust to the uncertainty around the current
history, which is induced by moves of Nature.
The prescribed actions should be suitable in 
any contingency of the unobservable state of nature. 
In the interaction between a single player and
Nature (or, equivalently, between two players with 
strictly conflicting
interests), the knowledge relevant to this task is of first order:
it regards only the set of possible contingencies,
that is, the information set.
 
Yet, in games among several players,
whether an action of one player  
is suitable or not at a particular
history may depend on whether another
player chooses a matching action at the same history. 
He, the one player, should thus be certain 
that she, the other player, would play the matching action, 
according to her commonly known strategy, 
which, however, responds to the observations 
she received on her own side. 
In other words, to avoid discoordination, 
he needs to know about what she knows about the current history.
In contrast to the one-player or the two-player conflict 
case, here it is relevant to consider
higher-order knowledge, i.e.,
knowledge about the knowledge of other players. 

The role of knowledge is particularly obvious in 
coordination games where the actions of players must agree at every
history. Formally, we call a \emph{consensus game}  
a coordination game with a set of actions 
that is common to all players, and where 
every move $(v, a, w) \in E$ in which two players $i \neq j$ 
disagree on their actions $a^i \neq a^j$ leads to 
a special sink state $\ominus$
from which no play is winning.
A~necessary condition for a strategy profile~$s$ to be winning in a consensus game is 
that, for every history~$\pi$ that follows~$s$, all components prescribe the
same action, that is, ~$s^i( \pi ) = s^j( \pi )$, for all players~$i,j$. 
When speaking about a winning strategies in such a game, 
we may therefore identify 
any strategy~$s^i$ of an individual player 
with the profile~$s$ of strategies where all components are equal to $s^i$, 
without loss of generality.


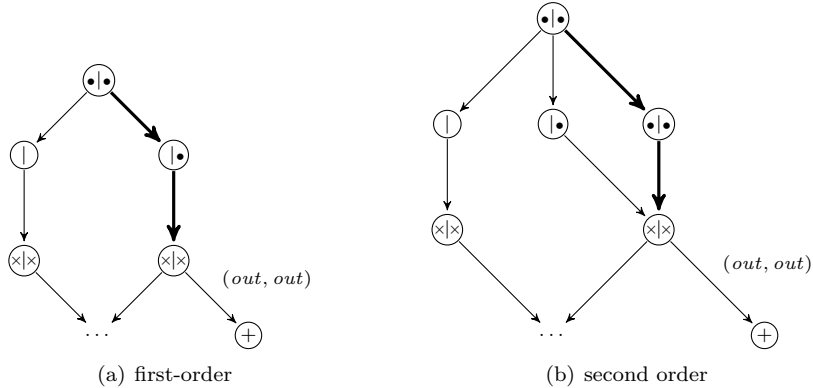
\begin{figure}[t]
\begin{center}
  \subfigure[first-order]{ 
    \label{fig:uncertainty-1a}
   \begin{tikzpicture}
      \node[state]       (0)              {$\gob|\gob$};
      
      \node[state]       (A) [below left of=0] {$~|~$};
      \node[state]       (B) [below right of=0] {$~|\gob$};
      \node[state]       (AA) [below of=A] {$\gout|\gout$};
      \node[state]       (BA) [below of=B] {$\gout|\gout$};
      \node[state, draw=none]       (G) [below right of=AA] {\dots};
      \node[state]       (Z) [below right of=BA] {$+$};
      \path (0) edge(A)
      edge [very thick] (B)
      (A) edge (AA)
      (B) edge [very thick] (BA);

      \path (AA) 
        edge (G);
        
        \path (BA) 
        edge (G)
        edge[bend right=0] node {$(\call,\call)$} (Z);

      \end{tikzpicture}
    }
    \hspace*{1cm}
    \subfigure[second order]{
      \label{fig:uncertainty-1b}
      \begin{tikzpicture}
        \node[state]                (0)              {$\gob|\gob$};
        
        \node[state]                (B) [below of=0] {$~|\gob$};
        \node[state]                (A) [left of=B] {$~|~$};
        \node[state]                (C) [right of=B] {$\gob|\gob$};

        \node[state]                (AA) [below of=A] {$\gout|\gout$};
        \node[state]                (BB) [below of=C] {$\gout|\gout$};
        \node[state, draw=none]    (H)  [below of=B] {};
        \node[state, draw=none]    (HH)  [right of=BB] {};
        \node[state]                (Z) [below of=HH] {$+$};
        \node[state, draw=none]                (G) [below of=H] {\dots};
        \path (0) edge(A)
        edge(B)
        edge [very thick] (C)
        (A) edge (AA)
        (B) edge (BB)
        (C) edge [very thick] (BB);
        
        \path (AA) 
        edge node[left] {} 
        (G);
        
        \path (BB) 
        edge node[right] {$$} 
        (G)
        edge[bend right=0] node {$(\call,\call)$}
        (Z);
        
      \end{tikzpicture}
    }
  \end{center}
  \caption{Lacking knowledge to coordinate} 
  \label{fig:uncertainty-1}
\end{figure}

Figures~\ref{fig:uncertainty-1} and~\ref{fig:uncertainty-2}
show examples of consensus games for
two players, Player 1~(he) and~2 (she), 
with actions $\pass$ and $\call$; unlabelled arcs represent
moves where both players choose $\pass$.
The observations $\bullet$, $\circ$, and $\times$ 
are indicated in split
states: he sees the left side, she the right side.
Apart from the unsafe sink~$\ominus$ that also collects 
the moves along any unrepresented action profiles,
there is a safe sink~$\oplus$, 
from which all plays are winning; the sinks are observable to both players. 
The dots on the bottom stand for an
arbitrary continuation. 
For each of these games, we consider the
situation where the actual history corresponds to the rightmost path,
marked by thicker arrows leading to the $\call$ state where playing
$(\call,\call)$ would lead to an immediate win.
Along the examples, we will discuss the question 
of whether the information that players have
at the marked history allows them to infer with certainty that
$(\call,\call)$ is a safe action.

In Figure~\ref{fig:uncertainty-1a}, 
at the marked history, Player~$2$ knows about being in the 
$\call$ state which requires the action $(\call, \call)$
to win. However, Player~$1$ cannot distinguish 
the current history from the one along the left path
$(\bullet|\bullet)(\circ|\circ)(\times|\times)$, 
where playing $\call$ would be losing.
So it occurs that at the current state
$(\call,\call)$ is the right joint move, but 
Player~$1$ lacks first-order knowledge about it, 
whereas Player~$2$ has the relevant first-order knowledge, 
but not the second-order knowledge to ascertain 
that Player~$1$ will play $\call$. 
At best, the players could coordinate on
$(\pass,\pass)$, based on their common knowledge that this 
leads to continuing the game, and not straight to the 
$\ominus$ sink.

In Figure~\ref{fig:uncertainty-1b}, both players know that they are in
the $\call$ state. 
Nevertheless, Player~$2$ is uncertain about 
whether Player $1$ knows it, 
because according to her observation, the current history may be 
$(\bullet|\bullet)(\circ|\bullet)(\times|\times)$, in which case,
just as in Figure~\ref{fig:uncertainty-1a}, Player~$1$ would
consider the history 
$(\bullet|\bullet)(\circ|\circ)(\times|\times)$ 
possible, and thus not play $\call$. 
So both players have first-order knowledge about
being at the $\call$ state, 
Player~$1$ even has the second-order knowledge that Player~$2$ 
knows it, and still 
they cannot coordinate with certainty, because Player~$2$ does not
know that Player~$1$ knows it.
Moreover, in Figure~\ref{fig:uncertainty-2a}, both players 
know that the play is at the $\call$ state
and each of them knows that they both know it. 
But Player~$1$ regards it as possible that
Player $2$ observed $\bullet\bullet\circ\times$,  
again raising the uncertainty of Figure~\ref{fig:uncertainty-1b}.   
\takeout{
  in which case she would consider that he
  may have observed $\bullet\circ\circ\times$ 
  leading him to consider that the current
  history was $(\bullet|\bullet)(\circ|\circ)(\circ|\circ)$, and hence
  not safe for $\call$.
}
Here, the reason why the players cannot coordinate is that 
Player~$1$ does not know that Player~$2$ knows
that Player~$1$ knows about being in a $\call$ state. 

The argument can be lifted to arbitrary levels of the knowledge
hierarchy. 
This is illustrated in Figure~\ref{fig:uncertainty-2b}, where
the loop around the observation $(\bullet|\bullet)$ may be
unravelled $n$ times to obtain an instance where coordination on the
winning action fails 
in spite of the players having mutual knowledge of order $n$ about
being in a state where this action is safe.

Indeed, the examples embody the \emph{coordinated
attack problem}, a parable that illustrates the intricacy 
of coordination via unreliable communication.
The story features two generals camped with their armies on two hills 
surrounding a fortification that they plan to attack. As
either one alone would lose the battle, 
they need to agree on attacking simultaneously. 
However, they can only communicate by sending messengers, 
which may be captured on the way. 
The challenge is to attain common knowledge about being in a state
where the attack can occur after a finite history of message
exchanges, given that whenever a general receives a message, 
he is uncertain of whether the sender knows that he received it.
Proofs that this cannot be achieved have been given for
different settings, e.g., by
Gray \cite{Gray78}, and
Halpern and Moses~\cite{HalpernMos90},
in the distributed-systems literature,  
and by Rubinstein \cite{Rubinstein89} in game theory.
In our setting, it is Nature that induces 
nondeterministically a possible loss of one message
between the two generals who would attack if, and only if, 
they had common knowledge of being in the $\call$-state.
The analyses put forward the paradigm that
that 
if common knowledge
is not attainable, then coordination is impossible in spite
of an arbitrarily high level of mutual knowledge.

As our examples suggest, already
in the simple case of consensus games 
with a safety condition (avoid the $\ominus$-sink), coordination
games with imperfect information are sensitive to common knowledge,
and thus vulnerable to problems caused by its inapproximability
through finite levels of mutual knowledge.
Still, one may argue that the problem of synthesising a joint
winning strategy does not invoke the reasoning process of
individual players. In the end, strategies only rely on 
first-order knowledge. 
Nevertheless, we will show in the remainder of this article 
that the problem of attaining common
knowledge about certain events\,---\,namely, the game state 
at the actual history\,---\,%
is relevant for solving coordination games, 
in the following sense.
\begin{enumerate}[(i)]
\item There exists a family of games that admit a
  solution if, and only if, 
  the players attain common knowledge
  about the state at a particular history.
\item Every game of infinite duration 
  where common knowledge of the actual state is attained infinitely
  often along every play can be solved effectively.
\end{enumerate}

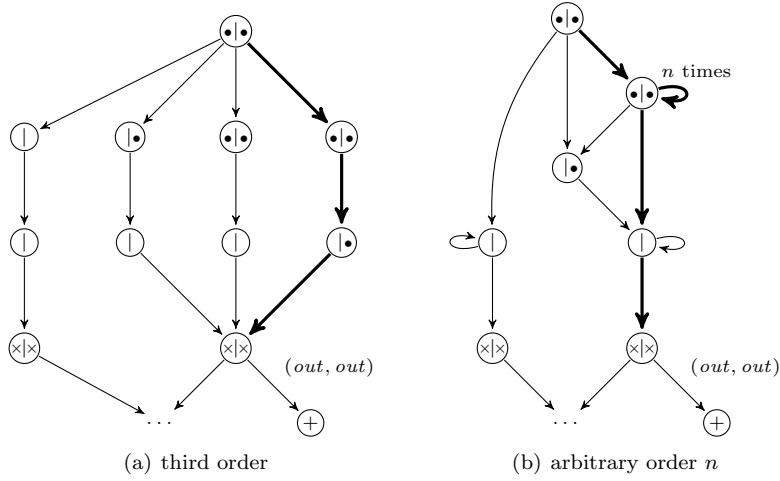
\begin{figure}[t]
\begin{center}
    \subfigure[third order]{
      \label{fig:uncertainty-2a}

      \begin{tikzpicture}
        \node[state]                (0)              {$\gob|\gob$};
        
        \node[state]              (A11) [below of=0] {$\gob|\gob$};
        \node[state]              (A01) [left of=A11] {$~|\gob$};
        \node[state]              (A00) [left of=A01] {$~|~$};
        \node[state]              (A12) [right of=A11] {$\gob|\gob$};
 
        \node[state]              (B11) [below of=A11] {$~|~$};
        \node[state]              (B01) [below of=A01] {$~|~$};
        \node[state]              (B00) [below of=A00] {$~|~$};
        \node[state]              (B12) [below of=A12] {$~|\gob$};
 
        \node[state]                (AA) [below of=B00] {$\gout|\gout$};
        \node[state]                (BB) [below of=B11] {$\gout|\gout$};   
        \node[state, draw=none]     (Z) [below left of=BB] {$\dots$};
        \node[state]                (G) [below right of=BB] {$+$};

        \path (0) edge(A00)
        edge(A01)
        edge (A11)
        edge [very thick](A12)
        (A00) edge (B00)
        (A01) edge (B01)
        (A11) edge (B11)
        (A12) edge [very thick] (B12)
         (B00) edge (AA)
        (B01) edge (BB)
        (B11) edge (BB)
        (B12) edge [very thick] (BB);
         
        \path (AA) 
        edge node[left] {} 
        (Z);
        
        \path (BB) 
        edge node[right] {$$} 
        (Z);
        
        \path (BB)
        edge node {$(\call,\call)$}
        (G);
        
      \end{tikzpicture}
    }
    \hspace*{.3cm}
    \subfigure[arbitrary order $n$]{
      \label{fig:uncertainty-2b}

      \begin{tikzpicture}
        \node[state]           (0)              {$\gob|\gob$};
        
        \node[state]           (A11) [below right of=0] {$\gob|\gob$};
        \node[state]           (A01) [below left of=A11] {$~|\gob$};

        \node[state]           (B1) [below right of=A01] {$~|~$};
        \node[state]           (A00) [below left of=A01] {$~|~$};

        \node[state]           (BB) [below of=B1] {$\gout|\gout$};   
        \node[state]           (AA) [below of=A00] {$\gout|\gout$};
        \node[state, draw=none] (G) [below left of=BB] {$\dots$};
        \node[state]           (Z) [below right of=BB] {+};

        \path (0) edge [bend right=20] (A00)
        edge(A01)
        edge [very thick] (A11)
        (A00) edge (AA)
        (A01) edge (B1)
        (A11) edge (A01) 
        edge [very thick] (B1)
        (B1) edge [very thick](BB);

        \path (AA) 
        edge (G);
        
        \path (BB) 
        edge (G);

        \path (A11) edge [very thick, loop right] node[above,
        xshift=4pt, yshift=2pt] {$n$ times} (A11);
        \path (A00) edge [loop left] node[above] {} (A00);
        \path (B1) edge [loop right] node[above] {} (B1);
        
        \path (BB) 
        edge node {$(\call,\call)$}
        (Z);

      \end{tikzpicture}
    }
  \caption{Coordinated attack}
  \label{fig:uncertainty-2}
  \end{center}
\end{figure}

\section{Common knowledge of the state}


Let~$G$ be a game graph, and let $\Omega$ be the set of
histories in~$G$. For a history~$\pi$, we denote by 
$\State( \pi ) \subseteq \Omega$ the set of 
all histories that end at the same state as~$\pi$. 
We say that the players 
attain \emph{common knowledge of the state} (\CKS) at
history~$\pi$, if $\State( \pi )$ is common knowledge at $\pi$.

To develop familiarity with the notion, 
we first show that attaining \CKS at a
particular history in a game graph
is equivalent to having a joint winning 
strategy in an associated coordination game, 
more precisely, a consensus game.
For simplicity, we detail the case where 
Nature controls all moves in the original game graph,
i.e., the players have only one, trivial action; 
the general case requires only notational changes.
 
Given a game~$G$, let  
$\pi$ be a history 
that starts at the initial state  $v_0$ and 
ends at some state $z \in V$. 
We construct a consensus game 
$\calG_\pi$ on the the disjoint union of $G$ and (the unravelling of) $\pi$ 
as follows.  
The players are the same as in~$G$, and they have  
a common set  $\{ \pass, \call \}$ of actions.
The state set of $\calG_\pi$ consists of 
copies of the states in $G$ and a fresh state~$\hat{\tau}$ for every history~$\tau$ 
in~$\pi$\,---\,note that 
the copy of the initial state~$v_0 \in V_0$ is distinct from 
the initial history~$\hat{v}_0$.
For each state copy, the observations get inherited from the corresponding 
state in $G$, 
and for each history from its last state.
Likewise, the moves from $G$ and along $\pi$ 
get inherited with the action label  
$\pass$ for all players (in consensus).
The initial-history state~$\hat{v}_0$ is
designated as the initial state of $\calG_\pi$, 
and we add $\pass$ moves from $\hat{v}_0$ 
to every successor of~$v_0$ in~$G$.
There is an unsafe sink $\ominus$ and a safe sink $\oplus$; the
winning condition requires to avoid~$\ominus$.
Finally, we add $\call$ moves (in consensus) to the safe sink $\oplus$
from the copy of state $z$ in $G$ and from 
the state $\hat{\pi}$ corresponding to the history~$\pi$.
Moves with any other action profile lead to the unsafe sink~$\ominus$, namely
any action profile that is not in consensus, the action $\call$ 
from any state other than $z$ or $\hat{\pi}$,  
and the action $\pass$ from $\hat{\pi}$.

\begin{prop}\label{prop:ck-game}
For a game graph $G$ and a history $\pi$,
the players have a joint winning strategy in the game $\calG_\pi$
if, and only if, they attain common knowledge 
of the state at $\pi$ in $G$.
\end{prop}

\begin{proof}
To see that winning in~$\calG_\pi$ 
implies attaining \CKS at $\pi$, 
suppose that there exists a joint winning strategy~$s$ 
in~$\calG_\pi$. Since this is a consensus game, 
we may assume that all components of 
$s$ are equal and can identify the profile~$s$ 
and its component strategies.
Now, let $C$ be the set of histories~$\rho$ in $G$
that follow~$s$ and are assigned $s( \rho ) = \call$. 
We argue that $C$ satisfies the
conditions of Theorem~\ref{thm:ck-evident} 
to witness that the players attain
\CKS at~$\pi$:
\begin{itemize}
  \item $C$ is a self-evident event, for each player~$i$: for every 
    history $\rho \in C$, any indistinguishable history 
    $\rho' \sim^i \rho$ follows~$s$ and is
    assigned the same action 
    $s( \rho') = s( \rho ) = \call$, which means $\rho' \in C$.
  \item $\pi \in C$: the action $\pass$ is losing at $\hat{\pi}$, 
    and since no winning strategy can avoid~$\hat{\pi}$,  
    we must have $s( \hat{\pi} ) = \call$. As $\pi$ and
    its copy ending at $\hat{\pi}$ are indistinguishable to all players, 
    it follows that
    $s( \pi ) = \call$. 
\item $C \subseteq \State(\pi)$: the action $\call$ 
is losing at all states except for $z$ and $\hat{\pi}$. 
As we assumed that $s$ is a winning strategy, 
all histories in $C$ must end at~$z$.
\end{itemize}

For the converse, assume that, at the history $\pi$ in~$G$,
the players attain common knowledge of $\State( \pi )$.
We define a function~$s$
that associates to any history~$\rho$ 
in~$G$ the action $s( \rho) := \call$ 
if $\rho$ has the same length as $\pi$ and 
$\State( \pi )$ is common knowledge
at $\rho$, and otherwise $s( \rho ) := \pass$.

First, let us verify that~$s$ 
is a valid strategy on the game graph~$G$, for each player~$i$. 
Notice that, if an event 
is common knowledge at a history $\rho$, then 
it is also common knowledge
at every indistinguishable history $\rho'\sim^i \rho$
(each history that is accessible
from $\rho'$ via a sequence of pairwise indistinguishable histories
is also accessible from $\rho$ via the same sequence preceded by
$\rho \sim^i\rho'$). 
In particular, whenever  
$s( \rho ) = \call$, for a history~$\rho$, that is,
when the event $\State( \pi )$ is common knowledge at $\rho$, 
it is also common knowledge at every history 
$\rho' \sim^i \rho$, hence 
$s( \rho' ) = \call$. 
Consequently $s( \rho ) = s(\rho')$ for every pair $\rho \sim^i \rho'$.

Now, we extend the strategy $s$ on $G$ 
to the graph of $\calG_\pi$ in the only consistent way, 
by assigning to every history that ends at
a state $\hat{\rho}$ corresponding to a history~$\rho$ in~$\pi$, the
action~$s( \rho)$ prescribed in~$G$.
We argue that $s$ is a winning strategy in the consensus game $\calG_\pi$:
The unsafe sink~$\ominus$ can only be reached 
by taking a wrong move in one of the
following two situations: either choosing $\pass$ at state $\hat{\pi}$, 
or choosing $\call$ at a state different from $z$ and $\hat{\pi}$.
The former situation is excluded, as
$s(\hat{\pi}) = s( \pi ) = \call$ holds by definition of~$s$.
The latter situation cannot occur either: On the one hand, 
at all histories $\rho$ in~$G$ with 
$s(\rho) = \call$, the players attain common knowledge of $\State(\pi)$,
so in particular, $\rho \in \State( \pi )$ 
(by definition of the knowledge operator, 
players can only know an event if it actually occurs). 
On the other hand, among the histories along the copy of $\pi$, 
only the one ending at~$\hat{\pi}$ 
has the same length as $\pi$, which is necessary to be asigned~$\call$.
In conclusion, all plays that follow $s$ are winning in~$\calG_\pi$.
\end{proof}

The argument illustrates that the need for
(common) knowledge about the actual game state
is a source of computational 
complexity in coordination games with imperfect
information. In general, there may be further sources.
One class of games, where the issue is exactly
whether the players attain 
common knowledge about reaching a certain state set, 
are \emph{consensus acceptors} 
investigated in~\cite{BvdB15}.
Essentially, these are
consensus games with a simple safety condition (avoid the sink~$\ominus$) 
where the players have only one nontrivial decision
in every play.
Which decision to take thus depends on
the common knowledge of the players about the actual state.
The games~$\calG_\pi$ from Proposition~\ref{prop:ck-game} 
as well as those represented in Figures~\ref{fig:uncertainty-1} 
of the previous section 
are examples of consensus acceptor games.
The complexity analysis 
for consensus acceptors in~\cite{BvdB15}
sheds light on the problem of attaining common knowledge of the state 
in an arbitrary game graph.  
For completeness, we reproduce the part of the analysis relevant for
our setting.

\begin{prop}\label{prop:PSPACE}
Given a game graph and a history $\pi$,  
the problem of deciding whether the players attain common knowledge of
the state at $\pi$ is $\PSPACE$-complete. 
\end{prop}

\begin{proof}

For membership, consider the procedure that
takes a game graph and a history $\pi$ as input, and iterates the
following loop: guess nondeterministically 
a player $i$ and a history $\rho \sim^i \pi$; 
accept if $\pi$ and $\rho$ end at different states, 
otherwise repeat with $\rho$ as the new value of $\pi$. 
\footnote{We adopt the convention that machines reject by looping; 
Hopcroft and Ullman \cite{HopcroftUll69} showed that
lower space bounds 
above $\NLOGSPACE$ do not change when
 dropping the halting assumption. 
} 
As any two indistinguishable histories have the same length, 
the procedure requires only linear space. 
It accepts if, and only if, there exists a sequence 
$i_1, \dots, i_k$ of players and 
histories 
$\rho_1 \sim^{i_1} \rho_2 \sim^{i_2} \dots \sim^{i_k} \rho_k$
with $\pi = \rho_1$ and $\rho_k \not \in \State( \pi )$, that is, 
if $\State(\pi)$ is not common knowledge at $\pi$. 
Hence, we have a nondeterministic $\PSPACE$ procedure for deciding the
complement problem which asks whether the players do \emph{not} have \CKS
common knowledge of the state 
at the given history. Since nondeterministic
and deterministic $\PSPACE$ are equal, 
it follows that the original 
problem can be solved in $\PSPACE$.
Even more, the argument shows that for any game 
there exists a linear-bounded automaton that recognises the set of 
histories at which the players attain~\CKS. 

To prove hardness, we describe a reduction from (the complement of) the 
\textsc{Corridor Tiling} problem.
Given a domino system $\calD = (D, E_V, E_H)$ with 
a frontier constraint $w \in D^\ell$ 
we construct, in polynomial time,
a game graph~$G$ for two players and a history~$\pi$, 
such that the players attain \CKS
at~$\pi$ if, and only if, the domino-problem instance $(\calD, w)$ 
is negative,
i.e., there does not exist a height $m \in \bbN$ 
such that the rectangle $Z(\ell, m)$ can be tiled
with~$w$ in the top row. 

The two players in~$G$
have one trivial action and their
observations correspond to the dominoes in~$D$.
The set of states consists of
 \emph{singleton} states 
$d \in D \setminus \{ \# \}$, \emph{pair} states $(d, b) \in E_V$, 
an initial state $v_0$, and 
two sinks $\oplus$ and $\ominus$.
At each singleton state $d$, 
both players receive the same observation $d$, 
whereas at each pair state $(d, b)$, the first player observes $d$ 
and the second player $b$;
at the initial state and the two sinks, 
both players receive the observation $\#$
corresponding to the vertical border domino.

The singleton states are connected by moves 
$d \to d'$ for every $(d, d')$ in $E_H$, and
the pair states by moves
$(d,b) \to (d',b')$ whenever $(d, d')$
and $(b,b')$ are in~$E_H$. 
From the initial state~$v_0$, there are moves to all 
singleton states $d$ with $(\#, d)$ in $E_H$, 
and all pair states $(d, b)$ with $(\#, d)$ and $(\#,
b)$ in $E_H$. 
Conversely, the sink~$\oplus$ is reachable from all 
singleton states $d$ with $(d, \#)$ in $E_H$, 
and from all pair states $(d, b)$ with $(d, \#)$ and 
$( b, \#)$ in $E_H$; the sink $\ominus$ is reachable only from
the singleton bottom-domino state~$\square$. 
Clearly, the game graph~$G$ can be constructed from~$\calD$ 
in linear time.

Note that any sequence 
$x = d_1, d_2, \dots, d_\ell \in D^\ell$ 
that forms a horizontally consistent row 
(omitting the borders) in a corridor tiling 
corresponds to a history 
$\pi_x = v_0 d_1 d_2 \dots d_\ell \oplus$ in the game
graph. For the special case of the bottom row~$\cbot$ 
with $d_1 = d_2 = \dots = d_\ell = \square$, 
apart of $\pi_\cbot$, we also have the history  
$\hat{\pi}_\cbot = v_0 d_1 d_2 \dots d_\ell \ominus$.
On the other hand, every history in $G$ that ends at a sink
corresponds either to one consistent row, 
in case Nature chooses a singleton state
in the first move, or to two rows, in case Nature chooses a pair
state.
Moreover, a row $x = d_1, d_2, \dots, d_\ell$ can appear 
below a row $y = b_1, b_2, \dots, b_\ell$ 
in a correct tiling 
if, and only if, 
there exists a history $\rho$  in $G$ such that 
$\pi_x \sim^1 \rho \sim^2 \pi_y$, namely 
$\rho = v_0 \, (d_1, b_1)\, (d_2,b_2) \dots (d_\ell,
b_\ell)\, \oplus$.

Now, we claim that  
the players attain \CKS at the history~$\pi_w$ 
corresponding to the frontier constraint~$w \in D^\ell$  
if, and only if, there exists no corridor tiling for the 
instance~$(\calD, w)$.
According to our observation, if there exists
a correct tiling of the corridor, then
there exists a sequence
of rows corresponding to histories $\pi_1, \dots, \pi_m$,
and a sequence of witnessing histories 
$\rho_1, \dots, \rho_{m-1}$ such that 
\begin{align*}
  \pi_w = \pi_{1} \sim^1 \rho_1 \sim^2 \pi_{2} \dots \sim^1
\rho_{m-1} \sim^2 \pi_{m} = \pi_\cbot.
\end{align*} 
However, the history $\pi_{\cbot}$ is indistinguishable from
$\hat{\pi}_\cbot$, for both players.
As these two histories end at different states, it 
follows that $\State( \pi_w )$
is not common knowledge at $\pi_w$. 
Conversely, if $\State( \pi_w )$ is not common knowledge at
$\pi_w$, then there exists a sequence of pairwise indistinguishable 
histories that leads from~$\pi_w$ to $\hat{\pi}_\cbot$, from which
we can extract a correct corridor tiling.

\textsc{Corridor Tiling} is $\PSPACE$-hard, 
according to Theorem~\ref{thm:corridor-complexity}. 
Thus, the reduction 
shows that
the problem of deciding whether the players attain
\CKS at a given history is $\textrm{Co\,-}\PSPACE$
hard, and since $\PSPACE$ is closed under complement, 
it is hard for $\PSPACE$.
\end{proof}

The above argument can be extended to prove that it is undecidable 
whether the players can ever attain \CKS along a given
set of plays in a game. 
We say that a history~$\pi$ is \emph{within}~a subset $S\subseteq V$ of states
if all states that occur in~$\pi$ belong to~$S$.

\begin{prop}\label{prop:ck-existential}
It is undecidable whether, for a 
given game graph with a designated subset~$S \subseteq V$ of states, 
there exists a nontrivial history within~$S$ at which the
players attain common knowledge of the state.
\end{prop}
\begin{proof}

We proceed by reduction from \textsc{Corridor Universality}:
Given a domino system $\calD$ 
with a designated subset~$\Sigma \subseteq D$, 
we construct a
game graph~$G$ with a subset~$S$ of states 
such that the players attain \CKS at some nontrivial 
history~$\pi$ within $S$
if, and only if, there exists a frontier constraint~$w \in \Sigma^\ell$ 
that does not allow a corridor tiling with~$\calD$;
actually, the sequence of observations along~$\pi$ will 
yield such a constraint.

The construction of~$G$ is as in the proof of 
Proposition~\ref{prop:PSPACE} except that
we take two disjoint copies of the game graph associated to~$\calD$
and identify the two copies of the initial state~$v_0$, 
and those of the sinks $\ominus$, $\oplus$, respectively. 
In this way, no player knows the
actual state of any history that does not reach a sink.
The set~$S$ consists of the initial state, the two sinks, and 
the states corresponding to singleton dominoes in~$\Sigma$ (from both copies). 

Now, every nontrivial history $\pi$ 
that ends at a sink
corresponds either to a single row or to a 
pair of horizontally consistent rows, 
depending on whether~$\pi$ proceeds through singletons 
or pair states. This holds by the same argument as in the proof of 
Proposition~\ref{prop:PSPACE}.
Likewise, it follows that the 
two players do not attain \CKS at a history $\pi$ 
if, and only if, the row corresponding to~$\pi$, or either one of the two rows, 
appears in the frontier of 
some correct corridor tiling with~$\calD$. 
For histories within $S$, 
when~$\pi$ is of the form $\#w\#$ for some (nonempty) word $w \in
\Sigma^*$, this means that the players do not attain \CKS if, and only
if,  there exists a correct corridor tiling with
$w$ in the frontier.

In conclusion, there exists a history within $S$ at which the 
players attain \CKS in~$G$
if, and only if, the \textsc{Corridor Universality} instance $(\calD, \Sigma)$
at the outset is negative\,---\,an undecidable problem, by
Theorem~\ref{thm:corridor-complexity}.
\end{proof}

The results of Latteux and 
Simplot stated in Theorem~\ref{thm:cs-domino}
describe a correspondence 
between context-sensitive languages and domino systems. 
Via our construction for proving the lower bounds in
Propositions~\ref{prop:PSPACE} and \ref{prop:ck-existential},
this correspondence is extended to game graphs, as 
detailed in~\cite{BvdB15} for the more general case 
of consensus acceptor games. 
For our setting, 
it implies that any context-sensitive language~$L \subseteq \Sigma^*$
can be translated into a game graph $G$ 
such that a word $w$ belongs to $L$
if, and only if, the players do not attain \CKS at a 
particular history~$\pi_w$  in $G$ which yields $w$ 
as an observation to both players. 
Conversely, the nondeterministic linear-space procedure witnessing the
upper bound in Propositions~\ref{prop:PSPACE} shows that
the set of histories at which the players do not attain \CKS is
recognisable by a nondeterministic linear-bounded automaton, and it is
hence context-sensitive.  
Thus, we can formulate the following
corollary which implies, in particular, 
that common knowledge of the state is
not a finite-state property in arbitrary games, i.e.,  
the set of histories at which the players attain
\CKS is not regular.

\begin{corol}[\cite{BvdB15}]\label{cor:ck-not-reg}
\begin{enumerate}[(i)]
\item 
In any game, 
the set of histories at which the players attain
\CKS forms a context-sensitive language. 
\item
For every context-sensitive language~$L \subseteq \Sigma^*$, 
we can construct a game in which an observation history belongs to~$L$
if, and only if, 
the players attain \CKS at the history.
\end{enumerate}
\end{corol}

\section{Recurring common knowledge}
\label{sec:recurring}
Let us now turn to the use of 
common knowledge in infinite plays.
We say that a play~$\pi$ allows 
for \emph{recurring} common knowledge of the state (\RCKS) if  
there are infinitely many histories in $\pi$ 
at which the players attain \CKS.
Likewise, we say that a game graph $G$, or a game over~$G$,  
allows for \RCKS
if this is true for every play in~$G$.

A \emph{knowledge gap} in a play $\pi$ is an interval 
$\igap{\ell}{t}$ with $t \ge \ell > 0$, 
such that the players do not attain \CKS 
in $\pi$ at any round in $\igap{\ell}{t}$.
The \emph{length} of the gap is $t - \ell + 1$.
Hence, a play allows
for \RCKS if the length of every knowledge gap in it 
is finite. 
The \emph{gap size} (for \CKS) of a play $\pi$ 
is the least upper bound on the length of knowledge gaps
in $\pi$.
Likewise, the gap size of a game (graph) 
is the least upper bound on the gap size
of its plays. 

Intuitively, 
the uncertainty of players about the game state 
progresses along knowledge gaps, and it vanishes
at every history at which \CKS is attained.
If we reindex the histories of a game~$G$ by forgetting 
any prefix history at which the players attain \CKS, the 
knowledge of players about the game state is preserved.
Concretely, for the case of private knowledge, 
let $\pi$, $\pi'$ be two histories   
at which the players attain \CKS,
and let~$v$ be the state at which they end.
Then, for any two continuations $\tau$, $\tau'$ 
and every player~$i$,  
we have
$\pi \tau \sim^i \pi' \tau'$ if, and only if, 
$v \tau \sim^i v \tau'$ in the game~$G$ with~$v$ as initial state.
The preservation for the case of common knowledge follows immediately; 
we formulate it for further reference.   

\begin{lemma}\label{lem:reindexing}
For a game~$G$, let~$\pi$ be a history 
at which the players attain \CKS, and let~$v$ 
be its last state.
Then, the players attain \CKS 
at a prolongation history $\pi\tau$ in~$G$ if, and only if, 
they attain $\CKS$ at history $v \tau$ in the 
game~$G$ with initial state $v$. 
\end{lemma}

Notice that for a play in an arbitrary game, 
the length of knowledge gaps 
may be unbounded, even if the play allows for \RCKS;
its gap size is then infinite.
Nevertheless, we show that, if a \emph{game} allows for \RCKS, 
then there exists a uniform, finite bound on the 
length of the knowledge gaps in its plays.

\begin{prop}\label{prop:ck-bounded}
If a game graph allows for 
recurring common knowledge of the state, 
then its gap size is finite.
\end{prop}

\begin{proof}
Let $G$ be a game graph that allows for \RCKS. 
Without loss of generality, we assume that all states are
reachable from the initial state~$v_0$.

For each state $v \in V$, we construct a tree $T_v$ 
that may be understood as the unravelling 
of $G$ from $v$, up to common knowledge.
The nodes of~$T_v$ correspond to the histories in $G,v$ 
that have no strict, nontrivial 
prefix at which the players attain \CKS. 
The edges are labelled with
action profiles and correspond to moves in $G$: 
for any history $\rho$ in the domain of~$T_v$  
at which the players do not attain \CKS, or for $\rho = v$, 
we have an edge 
$(\rho, a, \rho a w)$ whenever $(u, a, w) \in E$,
for the last state $u$ of $\rho$.
The leaves of $T_v$ thus correspond to the 
histories in $G,v$ at which
the players attain \CKS for the first time 
(not counting the initial history).
Finally, we associate to every history the observations 
of its last state.

Notice that each of the constructed trees
has finite branching and all its paths are finite, 
according to our assumption that all plays allow for 
\RCKS.
Hence by K\"onig's lemma, 
every tree in the collection ($T_v)_{v \in V}$ is finite. 
We claim that the maximal height of a tree in this 
collection is an upper bound for the length of 
knowledge gaps in the plays of $G, v_0$.

To show this, 
we construct a game graph $\tck{G}$ over the
disjoint union of all unravelling trees~$T_v$,
where we identify every leaf history
with the root of the tree associated to its last state.
Formally, in each tree $T_v$, we replace every edge 
$(\rho, a, \pi)$, where $\pi$ is a leaf history 
ending at $w$,
with an edge $(\rho, a, w)$ leading to the root of the tree~$T_w$.
This induces a natural bijection $h$ between histories of 
$G, v_0$ and $\tck{G}, v_0$, which
is also a bisimulation\,---\,clearly, the two game graphs have the same
infinite unravelling. 
\takeout{
Another way to view $\tck{G}$ 
is to consider the infinite unravelling of
$G,v_0$ and to 
identify, for each state $v \in V$, 
all 
histories of $G$ that end at $v$ and at which 
the players attain \CKS. 
}
The bijection $h$ preserves \CKS: 
By the reindexing argument of Lemma \ref{lem:reindexing},
the players attain \CKS at a history $\pi$ in $G, v_0$,
if, and only if, they attain \CKS at the image $h( \pi )$ in 
$\tck{G}, v_0$.
\takeout{
Indeed, the \emph{if} part holds in an even stronger sense: 
if two histories 
$\pi$ and $\pi'$ end at the same state in $G, v_0$ but 
this is not common knowledge, 
then the (images of the) histories end at
different states in $\tck{G}, v_0$.
}
As a consequence it follows that, on the one hand, every history 
in $\tck{G}, v_0$ at which the players attain \CKS ends at the root of 
some tree~$T_v$, and
on the other hand, for every knowledge gap, i.e., 
every sequence of consecutive 
histories $\pi^1, \pi^2, \dots, \pi^t$ in 
$G, v_0$ at which the players do not attain \CKS, the image 
$h(\pi^1), h(\pi^2), \dots, h(\pi^t)$ 
describes a sequence of consecutive histories in $\tck{G}$ 
that never visit the root of any tree~$T_v$. 
Hence, the length $t$ of such a sequence is bounded by 
the maximal length of a path in any of the trees $(T_v)_{v \in
  V}$. This concludes the proof.
\end{proof}

The insight that games with recurring common knowledge of the state 
have finitely bounded knowledge gaps  
allows us to conclude that these games are decidable
via a generic argument which, however, does not yield meaningful
complexity bounds.

\begin{theorem}~\label{thm:synthesis-decidable}
For games that allow for recurring common knowledge of the
state, with observable $\omega$-regular winning conditions,
\begin{enumerate}[(i)]
\item
it is decidable whether there exists a joint winning
strategy, and
\item 
if it is the case, there also
exists a finite-state winning strategy, 
which can be constructed effectively. 
\end{enumerate}
\end{theorem}

\begin{proof}
The argument relies on the tracking construction 
from~\cite{BKP11} that reduces the problem of solving 
coordination games with imperfect information for $n$ players against
Nature to that of solving a zero-sum game for two players 
with perfect information. 
The construction proceeds via
an unravelling process 
that generates epistemic models of the player's information along the
rounds of a play,
and thus encapsulates their uncertainty. 

This process
described as ``epistemic unfolding''
in the paper \cite[Section~3]{BKP11} is outlined as follows.
An \emph{epistemic model} for a game graph $G$ with the usual
notation, is a 
Kripke structure $\calK = (K, (Q_v)_{v \in V}, (\sim^i)_{1 \le i \le n})$ 
over a set $K$ of histories of the same length in in $G$, 
equipped with predicates $Q_v$ designating the histories that end in state 
$v \in V$ and with 
the players' indistinguishability relations $\sim^i$.
The construction keeps track of how the knowledge 
of players about the actual history is updated during a round,
by generating for each epistemic model $\calK$
a set of new models, one for 
each assignment of an action profile~$a_k$ 
to each history $k \in K$ such that  
the action assigned to any player~$i$  
is compatible with his knowledge, i.e.
for all $k, k' \in K$ with
$k \sim^i k'$, we have $a_k^i = a_{k'}^i$. 
The update of a model~$\calK$ with such an action 
assignment~$(a_k)_{k\in K}$ leads to a new, 
possibly disconnected epistemic model $\calK'$ 
over the universe 
\begin{align*}
  K' = \{k a_k w \mid k \in K \cap Q_v \text{ and } (v,a_k,w) \in
  E \},
\end{align*} 
with predicates $Q_w$ designating the histories $k a_k w \in K'$,
and with $k a_k w \sim^i k' a_k w'$ whenever
$k \sim^i k'$ in~$\calK$ and $w \sim^i w'$ in~$\calG$.
By taking the connected components of this updated model under 
the coarsening $\sim := \bigcup_{i = 1}^{n}\!\!\sim^i$, 
we obtain the 
set of epistemic successor models of $\calK$ in the unfolding. 
The tracking construction starts from 
the trivial model that consists only of the initial 
state of the game~$\calG$. 
By successively applying the update, it unfolds a tree 
labelled with  
epistemic models, which corresponds to a two-player game $\calG'$
of perfect 
information
where the strategies of one player translate into 
joint strategies of the grand coalition in~$\calG$
and vice versa, such that a strategy in~$\calG'$ is winning if and
only if the corresponding 
joint strategy in~$\calG$ is so~\cite[Theorem 5]{BKP11}.
 
The construction can be exploited algorithmically if the
perfect-information tracking of a
game can be folded back into a finite game. 
A homomorphism from an epistemic model 
$\calK$ to $\calK'$ is a function $f: K \to K'$ that preserves the
state predicates and the indistinguishability relations, that is, 
$Q_v(k) \Rightarrow Q_v (f(k)) $ and 
$k \sim^i k' \Rightarrow f(k) \sim^i f(k')$. 
The main result of~\cite{BKP11} shows that, whenever two nodes of the unfolded 
tree carry homomorphically equivalent labels, 
they can be identified without changing the 
(winning or losing) status of the game~\cite[Theorem 9]{BKP11}. 
This holds for all imperfect-information games 
with $\omega$-regular winning conditions that are observable. 
Consequently, the strategy synthesis problem is decidable for a
class of such games, whenever  
the unravelling process of any game in the class 
is guaranteed to generate only finitely many
epistemic models,  
up to homomorphic equivalence.

Let us now consider the tracking of a coordination 
game $\calG$ with observable $\omega$-regular winning condition that
allows for \RCKS.  
We claim that every history~$\pi$ where the players attain \CKS
corresponds to an epistemic model 
that is homomorphically equivalent to one with a single element
labelled with the (commonly known) state at which the history~$\pi$ ends. 
This is because, by our hypothesis of \CKS, 
in the $\sim$-connected component of any epistemic model 
containing~$\pi$, all histories end at the same state. 
On the other hand, when updating an epistemic model~$\calK$,
there are only finitely many successor models 
and each of them can be at most exponentially larger 
than $\calK$, for any fixed action space.
Accordingly, the number of updating rounds in which the 
models can grow is bounded by the gap size
of $\calG$, which is finite, 
according to Proposition~\ref{prop:ck-bounded}.

Therefore, every game with \RCKS has a finite tracking
quotient under homomorphic equivalence. 
By~\cite[Theorems 9 and 11]{BKP11}, this implies that the winner determination problem is decidable for such games, 
and finite-state winning strategies can be effectively synthesised 
whenever the players have a joint winning strategy.  
\end{proof}

\section{Characterisation via mutual knowledge}

Proposition~\ref{prop:ck-existential} 
leaves little hope for deciding
whether a game allows for \RCKS by checking that the property holds 
within parts of the game graphs or on individual plays.
In general, histories 
at which the players do not attain \CKS may be connected to 
arbitrarily long 
chains of indistinguishable histories that end at the same
state, before reaching one with a different end state
to witness the lack of \CKS.
Fortunately, there is a way around this obstacle. 
It turns out that in any game that allows for
\RCKS we can find, for each play~$\pi$, an associated play~$\pi'$
that aligns witnesses for all the histories in~$\pi$ that lack \CKS; 
in particular, whenever 
the players lack common knowledge of the state at some round in~$\pi$, 
there is one player that lacks first-order 
knowledge of the state in~$\pi'$. 
This will allow us to characterise games 
with recurring
common knowledge of the state as those where mutual knowledge of the
state is attained over and over again, along every play.

We say that the players 
attain \emph{mutual knowledge of the state} (\MKS) at
a history~$\pi$ in a game 
if $\State( \pi )$ is mutual knowledge at $\pi$, that
is, if all indistinguishable histories $\rho\sim^i\pi$ 
end at the same state as $\pi$, for all players~$i$.
A~play~$\pi$ allows for recurring mutual knowledge of the state 
(\RMKS) if the players attain \MKS at
infinitely many histories along $\pi$, and 
a game (graph) $G$ allows for \RMKS 
if all plays in~$G$ do.


The link between common and mutual knowledge is made 
by the notion of connected ambiguous histories.
We say that two histories, or plays, $\pi$ and $\pi'$
are \emph{connected} 
if there exists a sequence of  
histories or plays $\pi_1, \dots, \pi_{k}$ and a sequence of players 
$i_1, \dots, i_{k+1}$ such that 
\begin{align*}
  \pi \sim^{i_1} \pi_1 \sim^{i_2} \dots \sim^{i_{k}} \pi_{k}
  \sim^{i_{k+1}} \pi'.
\end{align*}
In the special case when two histories $\pi$ and $\pi'$ end at the
same state~$v$ and are 
connected via a sequence of histories that also end at~$v$, we say
that $\pi$ and $\pi'$ are \emph{twins}.
Clearly, the relations of connectedness and twins
are equivalences between histories. 
Moreover, if two histories~$\pi$ and $\rho$ that end at the same state~$v$
are connected or twins, then for every move $(v, a, v') \in E$ 
the prolongation histories $\pi a v'$ and $\rho a v'$ 
are in the same relation.

A history $\pi$ is \emph{ambiguous} if
the players do not attain \MKS at~$\pi$, 
that is, if 
there exists an indistinguishable history
$\rho \sim^i \pi$, for some player~$i$, 
which ends at a different state. 
In this case, we refer to $\rho$ is an \emph{ambiguity witness} for~$\pi$ 
and we say that $\rho, \pi$ is an \emph{ambiguous pair}.  
Notice that the  players do not attain  
\CKS at a history $\pi$ 
if, and only if, there exists
an ambiguous twin of~$\pi$.

Our goal is to show that every play in which the players do not attain
recurring common knowledge of the state 
is witnessed by one where they do
not attain recurring mutual knowledge of the state.
Towards this, we first prove that if the
players never attain \CKS in a play, 
there exists a witnessing play
in which all histories are ambiguous.

\begin{lemma}\label{lem:never-ck}
  For any game, if there exists a
  play $\pi$ along which the players never attain 
  \emph{common} knowledge of the state 
  (except for the initial state), then there 
  also exists a play $\pi'$ along which they never attain \emph{mutual}
  knowledge of the state. 
  Moreover, the plays $\pi$ and $\pi'$ are connected.
\end{lemma}

\begin{proof}

For an arbitrary game $G$ and a play $\pi = v_0 \,a_1 v_1\, \dots$,   
we consider the set $T_\pi$ of all histories $\tau$ in $G$ 
such that
every nontrivial prefix history of $\tau$ is ambiguous and 
connected to 
the history of the same length in~$\pi$.
As the set $T_\pi$ is closed under 
prefix histories, we can view it as  
a finitely branching tree. 
We wish to show that 
if the players do not attain \CKS along $\pi$, 
then every history in $\pi$ is connected to some history in 
$T_\pi$, and therefore~$T_\pi$ contains an infinite play in $G$ 
along which the players never attain $\MKS$. 

We prove a stronger property, for every 
history $\pi_\ell$ of lenght~$\ell \ge 1$ in $\pi$:
If the players do not attain \CKS along $\pi_\ell$ (except for the
trivial history), then for every ambiguous pair $\tau \sim^i \rho$
connected to $\pi_\ell$ 
there exists a pair $\tau' \sim^i \rho'$, 
such that 
\begin{enumerate}[(i)]
\item $\tau' \in T_\pi$ is a twin of $\tau$, and
\item $\rho'$ ends at the same state as $\rho$.
\end{enumerate}

For the base case with $\ell = 1$, 
if the players do not attain \CKS 
at the history $\pi_1 := v_0 \, a_1 v_1$, then 
there exist ambiguous histories connected to $\pi_1$,
and they all belong to $T_\pi$, because the only preceding
history is trivial. 
Hence, for any ambiguous pair $\tau \sim^i \rho$, 
already $\tau' = \tau$ and $\rho' = \rho$ witness the statement.

For the induction step, suppose the statement holds 
for $\ell \ge 1$ and 
assume that the players
do not attain \CKS up to (and including) the history
$\pi_{\ell+1}$ of length $\ell + 1$. 
In particular, this means that there exist ambiguous 
histories connected to $\pi_{\ell+1}$;
among these, let us pick an ambiguous pair 
$\tau a v \sim^i \rho c w$, with $v \neq w$. 
Due to perfect recall, we have $\tau \sim^i \rho$.

We distinguish two cases. (1) 
If $\tau$ and $\rho$ end at different states $v' \neq w'$, 
by induction hypothesis, there exists a twin $\tau' \in T_\pi$ of $\tau$ 
and a history $\rho' \sim^i \tau'$ that ends at $w'$. 
On the one hand $\tau' a v$ is a twin of $\tau a v$.
On the other hand, 
by definition of the observation function, 
$\tau' a v  \sim^i \rho' c w$. Since $\tau' \in T_\pi$ and $v \neq w$, 
this also implies $\tau' a v \in T_\pi$.
(2) Otherwise, suppose $\tau$ and $\rho$ end at the same state.
As the histories are connected to $\pi_\ell$, the players do not attain \CKS 
at~$\tau$. 
Hence, there exists an ambiguous twin $\tau'$ of $\tau$, and by induction  
hypothesis, we can choose~$\tau' \in T_\pi$. 
On the one hand, $\tau' a v$ is a twin of $\tau a v$.
On the other hand, as $\tau'$ and $\rho$ end at the same state, 
so $\tau' c w$ is a valid history in $G$, and we have $\tau' a v \sim^i \tau' c w$. 
Again, since $\tau' \in T_\pi$, and $v \neq w$ 
it follows $\tau' a v \in T_\pi$. 
This completes the induction argument.

In conclusion, 
for a play $\pi$ in which the players do 
not attain \CKS at any round, there exist
histories in $T_\pi$ that are connected to arbitrarily 
long histories of $\pi$. As the tree $T_\pi$ is finitely branching, 
it follows from K\"onig's Lemma that it 
has an infinite path $\pi'$. 
By construction, 
each nontrivial prefix of $\pi'$ is an ambiguous history, and it is  
connected to the history of $\pi$ of the same length. 
Hence,~$\pi'$ describes a play connected to~$\pi$ 
along which the players never attain
mutual knowledge of the state.
\end{proof}


We are now ready to formulate
our characterisation result that will be instrumental 
for the algorithmics of games with~\RCKS.

\begin{theorem}\label{thm:ck-char-mutual}
A game allows for recurring common knowledge of the state
if, and only if, it allows for recurring mutual 
knowledge of the state.
\end{theorem}

\begin{proof}
The \emph{only if} direction is trivial: common knowledge of an event
implies mutual knowledge. 

For the converse, let us consider a game $G$ 
that does not allow for \RCKS. Then, there exists a play~$\pi$ 
in which the players attain \CKS at some round $\ell$, but not at any 
later history. 
Accordingly, in the game $G, v$ starting from the 
(commonly known) state $v$ that is reached in round~$\ell$ of~$\pi$, 
there exists a play along which the
player never attain \CKS, except for the initial state.
Then, by Lemma~\ref{lem:never-ck}, there exists a
play~$\pi'$ in $G,v$ along which the players never 
attain mutual knowledge of the state.  
Furthermore, in the play that follows $\pi$ for
the first $\ell$ rounds and, upon reaching $v$, proceeds like $\pi'$, 
the players do not attain \MKS at the infinitely many
histories from round $\ell$ onwards.
Hence, the game $G$ 
does not allow for \RMKS, which concludes the proof.
\end{proof}

Before turning to algorithmic questions, 
let us state the following corollary of 
arguments from the proofs of 
Lemma~\ref{lem:never-ck} and 
Theorem~\ref{thm:ck-char-mutual},
which will be useful for bounding
the gap size of games
in Section~\ref{sec:decidability}.

\begin{corol}\label{cor:connect-ck-mk}
For any game~$G$, if the players 
do not attain common knowledge of
the state in a play $\pi$ along
a sequence of rounds
$\ell+1, \dots, \ell+t$, 
then there exists a play $\pi'$ in $G$ that is 
connected to $\pi$ and on which the players 
do not attain mutual knowledge of the state
along the rounds $\ell+1, \dots, \ell+t$.
\end{corol}

\begin{proof}
Let $G$ be a game graph and 
let $\pi$ be a play with the stated property, for some $\ell, t >
0$. We assume, without loss of generality, 
that the players attain \CKS at round $\ell$ in $\pi$.
For the game $G, v$ starting at the state $v$ reached in this round, 
we consider the suffix $\tau$ of
$\pi$ from round $\ell$ onwards, and construct 
the tree~$T_\tau$ of hereditarily ambiguous histories connected to $\tau$, 
as in the proof of Lemma~\ref{lem:never-ck}. 
The induction argument from the
proof then shows that the history of length~$t$ in~$\tau$ is connected to some
ambiguous history $\tau' \in T_\tau$. 
The histories of $\tau'$
from round $1$ to $t$ are ambiguous and each of them
is connected to the
history of the same length in~$\tau$. 
Hence, the play $\pi'$ that follows~$\pi$ for the 
first~$\ell$ rounds, 
then proceeds like~$\tau'$ for~$t$ rounds, and then again
follows $\pi$ satisfies the required properties: 
$\pi'$ is connected to $\pi$ and the players
do not attain \MKS
along the rounds $\ell+1, \dots, \ell+t$.
\end{proof}

\section{Recognising recurring mutual and common knowledge}
\label{sec:decidability}

An automaton for recognising
the plays that allow for \RMKS
could easily be designed using the powerset construction 
described by Reif~\cite{Reif84}
for solving one-player games with imperfect information.
This would yield a $\PSPACE$-procedure for deciding whether 
a game allows for \RMKS and thus for \RCKS. 
To obtain a sharper complexity bound, 
we will show that ambiguity witnesses along a play can be
represented efficiently, by a tree of very low width, 
which allows to reduce the complexity to $\NLOGSPACE$.

Let us fix an arbitrary game graph $G$.
A \emph{fork tree} for a play $\pi$ is a prefix-closed set~$T$ of 
histories that contains, for every~level $\ell \ge 0 $, 
\begin{enumerate}
\item[(i)]\label{it:fork1} 
  the history $\pi_\ell$ of $\pi$ in round $\ell$, and 
\item[(ii)]\label{it:fork2}
  at most one history~$\rho_\ell \neq \pi_\ell$ with
  $\rho_\ell \sim^i \pi_\ell$, for some player~$i$.
\end{enumerate}
A fork tree~$T$ is \emph{complete}, if it additionally satisfies, 
for every level $\ell$:
\begin{enumerate}
\item[(iii)]\label{it:fork3}
  if $\pi_\ell$ is ambiguous, then~$T$ contains an ambiguity 
  witness $\rho_\ell$ of $\pi_\ell$.
\end{enumerate}
We can view fork trees
as induced subtrees in the unravelling of~$G$  
that contain~$\pi$ as a central branch and 
have width at most two, that is, at most two elements on each level.
For convenience, we let~$\rho_\ell$ refer to $\pi_\ell$ whenever~$T$ 
contains only~$\pi_\ell$ at level~$\ell$.      
In case $\pi_\ell$ and $\rho_\ell$ end at different states, we say that
the level $\ell$ is a \emph{doubleton}, 
else it is a \emph{singleton}. 

If we consider an  
arbitrary family of ambiguity witnesses 
to the histories of a play,
the subtree induced in the game unravelling
can have unbounded width.
Nevertheless, the following lemma states that every play~$\pi$ 
admits a complete family of 
witnesses~$\rho_\ell$, one for each of its ambiguous 
histories $\pi_\ell$, that forms
a fork tree.

\begin{lemma}\label{lem:complete-fork}
For every play in an arbitrary game there exists  
a complete fork tree.
\end{lemma}

\begin{proof}

It is convenient to extend the notion of ambiguity 
witness to knowledge gaps in histories.
For a history~$\pi$ and an interval $\igap{\ell}{t}$,
we say that a history~$\pi'$ is an ambiguity
witness 
\emph{along the gap $\igap{\ell}{t}$} 
if $\pi$ and $\pi'$ have length at
least $t$, and~$\pi'_r$ is an ambiguity witness for $\pi_r$, 
for every round $\ell \le r \le t$.
Likewise, for a play~$\pi$, 
we say that a play~$\pi'$ is an ambiguity witness
\emph{from round~$\ell$ onwards} 
if $\pi'_r$ is an ambiguity witness for $\pi_r$ 
for every $r \ge \ell$.

Now, consider an arbitrary game $G$ and a play $\pi$.
By induction on the number of rounds $\ell$, we 
construct a finite or infinite sequence  
of trees $T_\ell$
that satisfy the fork-tree conditions~(i) and~(ii) 
for the first~$\ell$ levels, 
and, in addition, the following
strengthening of the completeness condition (iii) for the last level~$\ell$: 
\begin{enumerate}[(i)]
\item[(iii)${}^*$]
If, for some
$t \ge \ell$, there exists a history in $G$ that is 
an ambiguity witness 
for $\pi$ along the gap $\igap{\ell}{t}$, 
then there also exists a prolongation history 
of~$\rho_\ell$
that is such a witness.
\end{enumerate}
In particular, this implies that whenever the history $\pi_\ell$ is
ambiguous, the level~$\ell$  in $T_{\ell}$ is a doubleton. 

Each tree $T_{\ell+1}$ is finite and 
extends its predecessor $T_{\ell}$ by one
level, except if the sequence ends at some stage $\ell+1$, 
in which case $T_{\ell+1}$ extends $T_\ell$ 
with the (infinite) prolongation of $\pi_\ell$ to $\pi$ 
and with a prolongation play $\rho$ of either~$\pi_\ell$ or~$\rho_\ell$
that is an ambiguity witness for $\pi$ 
from round $\ell$ onwards. 
 
For the base case, we take the tree $T_0$ consisting only 
of the initial history~$v_0$. 
For the induction step, suppose that a tree
$T_\ell$ with $\ell$ levels satisfying the conditions 
(i), (ii), and (iii)${}^*$ has been constructed.
To extend it to $T_{\ell+1}$, we
look at the set $R$ 
of histories $\tau$ that prolong either $\pi_\ell$ or $\rho_\ell$, 
and are ambiguity witnesses for $\pi$ along the gap 
$\igap{\ell + 1}{t}$ up to the length $t$ of $\tau$.
Now we distinguish three cases.
(1) If $R$ is empty, 
we set $\rho_{\ell+1} := \pi_{\ell+1}$,  that is,
$\ell+1$ is a singleton level.
(2) If $R$ is nonempty, but finite, we pick a history $\tau \in R$ 
of maximal length, and add 
$\rho_{\ell+1} := \tau_{\ell+1}$ together with $\pi_{\ell + 1}$ as a
new level to~$T_{\ell}$. 
(3) Finally, if $R$ is infinite, there exists an infinite play $\tau$ 
in $G$ such that all its histories
from round $\ell$ onwards are in $R$. This follows from K\"onig's
Lemma, since the histories in $R$ form an infinite tree that is 
finitely branching 
(indeed, a subtree of the unravelling of $G$).
In this case, we add the histories $\pi_r$ and 
$\rho_r := \tau_r$, for all levels $r > \ell$ and
terminate the sequence with this infinite tree $T_{\ell+1}$.

In any case, $\rho_{\ell+1}$ is a history in $G$ and
is indistinguishable from $\pi_{\ell+1}$ which is also contained on 
level $\ell+1$. 
Condition~(iii)* holds trivially in case (3), we shall verify 
that it is also maintained in case~(1) and (2).

For case (1) assume, towards a contradiction, 
that $R$ is empty and there exists a history $\pi'$ of length
$\ell+1$ that is an ambiguity witness for $\pi_{\ell+1}$.
If $\pi'_\ell$ ends at the same state as $\pi_\ell$, 
then the last action-state pair $(a'_{\ell+1} v'_{\ell+1})$ of $\pi'$ yields
a prolongation $\tau = \pi_\ell a'_{\ell+1} v'_{\ell+1}$ 
that should be included in~$R$, a
contradiction. 
Else, if $\pi'_\ell$ ends at a different state than $\pi_\ell$, 
by perfect recall, $\pi'_\ell$ is an ambiguity witness for 
$\pi_{\ell}$ along the gap $\igap{\ell}{\ell+1}$, 
which, by induction hypothesis, implies that there
also exists such a witness 
that prolongs $\rho_{\ell}$ and is thus contained in~$R$, again
in contradiction to our
assumption that $R=\emptyset$.

For case (2), consider a history $\pi'$ of length $t > \ell$ 
that is an ambiguity witness for $\pi$ 
along the gap $\igap{\ell+1}{t}$.
We claim that there also exists a 
prolongation of $\rho_{\ell+1}$ 
with this property. There are two situations to distinguish:
If $\pi'_\ell$
reaches the same state as~$\pi_\ell$, then the history $\pi''$ that
follows $\pi$ until round $\ell$ and then continues like $\pi'$ 
belongs to $R$, 
and is at most as long as 
the witness $\tau$ chosen to construct $\rho_{\ell+1}$. Hence,
$\tau$ prolongs $\rho_{\ell+1}$ and is an ambiguity
witness for $\pi$ along the gap $\igap{\ell+1}{t}$. 
Otherwise, if $\pi'_\ell$ reaches a different state than $\pi_\ell$, 
then, by perfect recall, we have $\pi'_\ell \sim^i \pi_\ell$, 
for some player~$i$, and hence $\pi'$ is
already an ambiguity witness for $\pi$ along the gap 
$\igap{\ell}{t}$.  
By induction hypothesis, there exists an
ambiguity witness $\pi''$ for $\pi$ 
along the gap $\igap{\ell}{t}$ that prolongs
$\rho_\ell$. Hence, $\pi'' \in R$ and, as  
the history $\tau \in R$ 
chosen 
to construct $\rho_{\ell+1}$ is of maximal 
length,~$\tau$ prolongs
$\rho_{\ell+1}$  and is also an ambiguity
witness for $\pi$ along the gap $[\ell + 1, t]$.

Clearly, each tree $T_{\ell}$ constructed along the induction
satisfies the conditions of a complete fork tree  and 
agrees with its successor $T_{\ell+1}$, up to level $\ell$.
In conclusion, the sequence converges and the 
infinite tree $T$ obtained at the limit is a
complete fork tree for $\pi$.
\end{proof}

Fork trees for a fixed play~$\pi$
can be represented by $\omega$-words. 
We say that a word $\tau \in V(AV)^\omega$ 
is a \emph{fork sequence} for a play~$\pi$
if it starts with $\tau_0 = v_0$ and 
there exists a fork tree~$T$ for~$\pi$ such that
$\tau_\ell$ is 
the last action-state pair of $\rho_\ell$ in $T$, 
for every $\ell>0$.
In the following we construct, for any arbitrary game, 
an $\omega$-word automaton that takes as input an (infinite) 
play~$\pi$ in~$G$ and 
guesses in every run a fork sequence for~$\pi$. 
The automaton is equipped with a co-B\"uchi acceptance condition:
a run is accepting if it visits only 
finitely many non-accepting states. 
For background on the automaton model, 
we refer to the survey~\cite[Chapter 1]{GraedelThoWil02}.

\begin{prop}\label{prop:co-buechi}
For any game with $m$ states, 
the set of plays 
that do not allow for recurring mutual knowledge of
the state is recognisable by a
nondetermistic co-B\"uchi automaton with $m^2$ 
states.
\end{prop}

\begin{proof}
Let us fix an arbitrary game graph $G$. 
We construct an $\omega$-word automaton~$A$
with co-B\"uchi acceptance condition 
that recognises the set of histories~$\pi$ in~$G$,
for which there exists a fork tree with only 
finitely many singleton
levels. 
To witness this, 
the automaton guesses non-deterministically a 
fork sequence $\tau$ 
for $\pi$ and accepts if the states at
$\tau_\ell$ and $\pi_\ell$ are different,  
for all but finitely many rounds $\ell$.

The states of the automaton are pairs of game states from $V$: the
first component keeps track of the input play,
the second one is used for guessing the fork sequence $\tau$.
The transition function 
ensures that the two components evolve according to the moves
available in the game graph and that the current input symbol 
yields the same observation as the second component 
to some player~$i$. 

Concretely, the co-B\"uchi automaton $A$ is defined over the input
alphabet  $A \times V$
on the state set $V \times V$ 
with initial state $(v_0, v_0)$
and transitions from state $( u, u')$ 
on input $(a, v)$
to state $(v, v')$ whenever $(u, a, v ) \in E$ 
and $\beta^i( v' ) = \beta^i( v )$, for some player~$i$, 
and either $(u', a', v') \in E$ or $(u, a', v') \in E$, 
for some action $a' \in A$ with $a'^i = a^i$ for this player.
The set of final states is $Q \setminus \{ (v, v)~|~ v \in V\}$;
the automaton accepts an infinite input word 
if all states that occur infinitely often in a run are 
final.

We claim that an input word 
$\pi \in V(AV)^\omega$ is accepted by~$A$ 
if, and only if,~$\pi$ corresponds
to a play in~$G$, and the players never
attain mutual knowledge of the state along $\pi$,
from some round onwards. 

For the \emph{if} direction, consider a play $\pi$ along which the
players never attain mutual
knowledge of the state from some round onwards. 
By Lemma~\ref{lem:complete-fork},  
there exists a complete fork tree $T$ for $\pi$, in which
all but finitely many levels are doubletons.
Let~$\tau$ be the fork sequence
associated to $T$. Then, 
the sequence $((\pi_\ell, \tau_\ell))_{\ell < \omega}$ 
describes a run of~$A$ on input~$\pi$
in which non-final states $(v, v)$ occur only at the finitely many 
positions $\ell$ corresponding to singleton levels in~$T$, thus
witnessing that $\pi$ is accepted.

For the converse, inputs that do not correspond to
histories in $G$ are rejected, by construction of $A$. 
Furthermore, if an input word $\pi$ corresponds
to a play with infinitely many histories $\pi_\ell$ 
at which the players attain mutual knowledge of the
state, then every run of
the automaton visits a non-final state 
whenever such an input prefix 
$\pi_\ell$ is read. 
As this occurs infinitely often, 
the input~$\pi$ is rejected.
\end{proof}

\begin{theorem}\label{thm:ck-decidable}
The problem of whether a game graph allows 
for recurring common knowledge, or equivalently, 
recurring mutual knowledge of the state 
is $\NLOGSPACE$-complete. 
\end{theorem}

\begin{proof}
According to the Characterisation Theorem~\ref{thm:ck-char-mutual},
a game graph~$G$ 
allows for recurring common knowledge of the state if, and only if, 
it allows for recurring mutual knowledge of the state. 
Our problem thus reduces
to checking whether the language recognised by the 
co-B\"uchi automaton~$A$ constructed for $G$ in
Proposition~\ref{prop:co-buechi} is non-empty. 
It is well known that the
non-emptiness test for co-B\"uchi automata 
is in $\NLOGSPACE$ (see, for instance, Vardi and Wolper~\cite{VardiWol94}). 

Concretely, a nondeterministic procedure can 
guess a run of $A$ that 
leads to a cycle included in the set of final states. 
This requires only pointers to three states of the automaton: 
two for the current transition and one for storing a state 
to verify that a cycle is formed. As each state of the automaton is
formed by two states of the game, the overall space requirement 
is logarithmic in the size of the game graph~$G$.
Accordingly,
the problem of determining whether a game graph allows for common
knowledge of state is in $\NLOGSPACE$.

Hardness for $\NLOGSPACE$ follows via a straightforward
reduction from directed graph acyclicity, 
shown to be $\NLOGSPACE$-hard by 
Jones in~\cite{Jones75}: Given a directed graph $G$, we
construct a game graph $G'$ for one player 
by taking two 
disjoint copies of~$G$ and assigning all non-terminal nodes 
with the same observation; each terminal node 
is assigned with a distinct observation and equipped with a self-loop.
Finally, we add a fresh initial state to $G'$, 
with moves to all other states.
Clearly, the game graph $G'$
can be constructed using logarithmic space, and 
the player has recurring (mutual, common) knowledge of the state 
in~$G'$ if, and only if, the directed graph $G$ is acyclic.  

This shows that the problem 
of determining whether a game graph allows for common
knowledge of the state, or equivalently, for mutual knowledge of the
state, is $\NLOGSPACE$-complete.
\end{proof}

\begin{theorem}\label{thm:ck-bounded}
The gap size of any game with $m$ states 
that allows for recurring common
knowledge of the state is bounded by $m^2$.
\end{theorem}

\begin{proof}
Consider a game~$G$ with $m$ states that allows for \RCKS.
Towards a contradiction, suppose that in $G$ there exists a
 play with gap size greater than
$m^2$, that is, the players do not attain \CKS along
a sequence of consecutive rounds $r, \dots, r + m^2$, for some $r$. 
Due to Corollary~\ref{cor:connect-ck-mk}, 
there also exists a play $\pi$ in $G$ 
such that the players do not attain 
\MKS in $\pi$ along
these rounds.
Let $T$ be a complete fork tree for $\pi$, according to
Lemma~\ref{lem:complete-fork}, and let $\tau$ be 
the associated fork sequence. 

As~$G$ allows for \RMKS, 
the automaton $A$ constructed in
Proposition~\ref{prop:co-buechi} recognises the empty language, 
in particular it rejects the run on $\pi$
described by $(\pi, \tau)$. 
But $A$ has at most $m^2$ states, so there must 
be a cycle in the transition graph that is visited by this run, 
say from position $\ell \ge r$ to $t \le r+m^2$.
Along the interval $\igap{\ell}{t}$,    
the players do not attain \MKS in $\pi$, therefore 
the corresponding levels
in the fork tree $T$ are doubletons, 
and the states on the cycle visited 
in the run $(\pi, \tau)$ from position $\ell$ to $t$
are final.

Consider now the sequences $\pi'$, and $\tau'$
that follow $\pi$ and $\tau$, respectively, 
until position $t$ and then loop from
$\ell$ to $t$ forever. 
Then, the pair $(\pi', \tau')$ describes
a run in~$A$ that eventually cycles through final states, 
hence, the input $\pi'$ is accepted. 
But this means that $\pi'$ is a play in $G$ that 
does not allow for~\RMKS, 
in contradiction to our assumption that all plays in $G$ 
allow for \RCKS.
\end{proof}

We observe that 
the quadratic bound on the gap size
is tight. 
Consider, for instance, the game graph $G_m$ 
for one player with two observations, black and white, 
depicted in Figure~\ref{fig:quadratic}, 
for an arbitrary number~$m > 1$. 
There is only one bit of uncertainty induced by the choice 
of Nature at the
initial state, where it can either move up, 
into the cycle with $m-1$ white states followed by
a black one, or down, 
to the path consisting of~$m$ white states with selfloops, 
each followed by a black state, except for the last
one which leads to the black state on the cycle. 
Consider the play~$\pi$ where Nature moves into the cycle 
(and stays there forever). 
Along~$\pi$, every nontrivial history up to 
round $m^2$ is indistinguishable from the one where Nature
moves initially down to the path and loops on each white state
precisely $m-1$ times.
For the first $m^2$ rounds in~$\pi$, the player does therefore 
not know the current state, which means that
the gap size of the game is at least~$m^2$. 
On the other hand, notice that 
all histories that are distinguishable from $\pi$
are non-ambiguous, and that from round $m^2 + 1$ onwards, any
history that is indistinguishable from $\pi$ 
leads to the same state as~$\pi$ itself. 
Accordingly the game graph $G_m$ with $3m$ states 
allows for \RCKS  and its gap size is $m^2$.

\begin{figure}[t]
\begin{center}
      \begin{tikzpicture}
        \node[state] (0)  {$~$};
        \node[state] (A1) [above right of=0] {$~$};
        \node[state, draw = none] (A1x) [right of=A1] {$~$};
        \node[state] (A2) [right of=A1x] {$~$};
        \node[state, draw = none] (A2x) [right of=A2] {$~$};
        \node[state, draw=none] (Adots0) [right of=A2x] {}; 
        \node[state, draw=none] (Adots1) [right of=Adots0] {};
        \node[state] (An) [right of=Adots1] {$~$};
        \node[state, draw=none] (Anx) [right of=An] {$~$};
        \node[state] (B1) [below right of=0] {$~$};
        \node[state] (B1x) [right of=B1] {$\gob$};
        \node[state] (B2) [right of=B1x] {$~$};
        \node[state] (B2x) [right of=B2] {$\gob$};
        \node[state, draw=none] (Bdots0) [right of=B2x] {$$};
        \node[state, draw=none] (Bdots1) [right of=Bdots0] {$$};
        \node[state] (Bn) [right of=Bdots1] {$~$};
        \node[state] (z) [above right of=Bn] {$\gob$};
        
        \path (0) edge (A1) edge (B1);
        \path (A1) edge (A2);
        \path (A2) edge (Adots0);
        \path (Adots0) edge[dotted,very thick,shorten >=.3em,shorten <=.3em,-] (Adots1);
       \path (Adots1) edge (An);
        \path (An) edge (z);
        \path (z) edge[bend left=10] (A1);
        \path (B1) edge[draw =none, bend left=5]
        node[above,pos=.7,-] {path of length $2m-1$}  (Bn);
        \path (B1) edge (B1x) edge [loop above] (B1);
        \path (B1x) edge (B2);
        \path (B2) edge (B2x) edge [loop above] (B2);
        \path (B2x) edge (Bdots0);
        \path (Bdots0) edge[dotted,very thick,shorten >=.3em,shorten <=.3em,-] (Bdots1);
           \path (Bdots1) edge (Bn);
           \path (Bn) edge (z) edge [loop above] (Bn);
           \path (A1) edge[draw =none, bend left=5]
           node[above,pos=.7,-] {cycle of length $m$}  (An);

      \end{tikzpicture}
  \caption{A game graph with 3m states and gap size $m^2$}
  \label{fig:quadratic}
  \end{center}
\end{figure}
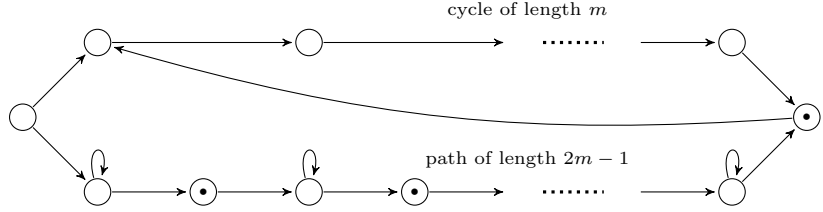

One consequence of Theorem~\ref{thm:ck-bounded} is that 
the knowledge hierarchy 
for any game of size~$m$ that allows for \RCKS collapses to
the level~$|B|^{m^2}$: 
By the Reindexing Lemma~\ref{lem:reindexing}, 
whenever two histories are connected, they are connected via 
a sequence of histories that differ only on the last knowledge gap. 
As the gap size is bounded by~$m^2$, 
there are at most 
$|B|^{m^2}$ different observation histories along a gap. 
Hence, the players attain common knowledge about an event 
$F \subseteq \Omega$ at a history~$\pi$ if, and only if, they attain \MKS of 
order~$|B|^{m^2}$ about~$F$. In particular, this holds for the event 
$\State( \pi )$. 

For any game,  
an automaton that recognises the set of histories 
at which the players attain
\MKS of a fixed order can be constructed  
as in the proof of~Proposition~\ref{prop:co-buechi}.
Accordingly, for any game that allows for \RCKS, 
the set of histories at which the players attain \CKS is regular.
This is in contrast with the general situation of games 
that may not allow for \RCKS where, 
by Corollary~\ref{cor:ck-not-reg},
the set of histories at which the players attain \CKS 
form arbitrary context-sensitive sets.

\section{Strategy synthesis}\label{sec:synthesis}

We are now ready to establish complexity 
bounds for the basic algorithmic questions on 
games with recurring common knowledge of the state. 
Our analysis focuses on the canonical case of 
parity condition. At the end of the section, we 
explain how the results apply to  
observable $\omega$-regular conditions.

\begin{theorem}\label{thm:ck-complexity}
For games that allow for recurring common knowledge of the
state, with parity winning conditions,
\begin{enumerate}[(i)]
\item \label{itm:decision-complexity}
the problem of deciding whether there exists a joint winning strategy 
is $\NEXPTIME$-complete;
\item \label{itm:construction-complexity}
if joint winning strategies exist, 
there also exists a winning profile of
finite-state strategies of at most exponential size, which 
can be synthesised in $2\textrm{-}\EXPTIME$. 
\end{enumerate}
\end{theorem}

The lower bound
for the decision problem~(\ref{itm:decision-complexity})
follows from $\NEXPTIME$-hardness 
of the corresponding problem 
for two-player reachability or safety games of finite horizon.
These are games where the underlying graph is acyclic except for having 
self-loops at observable sinks\,---\,hence, the simplest examples of
games that allow for \CKS.  
The original proof, due to 
Azhar, Peterson, and Reif~\cite[Section 5]{AzharPetRei01}, 
is by reduction from the time-bounded halting problem 
via a variant of QBF with dependency quantifiers.

For the sake of completeness, we outline a direct reduction from the
\textsc{Exp-Square Tiling} problem to synthesis problem for safety
games with finite horizon, 
similar to the one described by
Bernstein, Zilberstein, and Immerman 
in~\cite{BernsteinetAl02} for
decentralised planning in partially
observable Markov decision processes of finite horizon.
 
Given a domino system~$\calD$ and the logarithm~$\ell$ of the 
square size, 
we construct a two-player game with the following scenario:
First, Nature sends 
to each player~$i$ privately a pair 
$(x^i,y^i) \in Z( 2^\ell, 2^\ell)$ 
of coordinates in binary encoding over $\ell$ bits, such that either
\begin{enumerate}[(i)]
\item\label{itm:equal} $(x^2, y^2) = (x^1, y^1)$, or 
\item\label{itm:incx} $(x^2, y^2) = (x^1 + 1, y^1)$, or 
\item\label{itm:incy} $(x^2, y^2) = (x^1, y^1 + 1)$.
\end{enumerate}
Then, each player~$i$ produces a domino~$d^i$.  
The play is winning if the produced dominoes
are consistent with the relative position of the recieved 
coordinates, that is, if $d^1 = d^2$ in case~(\ref{itm:equal})~, 
$(d^1, d^2) \in E_H$ in case~(\ref{itm:incx}), and
~$(d^1, d^2) \in E_V$ in case~(\ref{itm:incy}). This can be formulated
either as an observable reachability or safety condition: reach~$\oplus$ or
avoid $\ominus$.
If a tiling of the exponential square exists, then the 
strategy to produce the domino placed at
the received coordinates guarantees a joint win.
Conversely, any winning strategy can be turned into
a correct tiling of the exponential square.  

The game is of finite horizon: After~$\ell + 1$ rounds ($\ell$ for observing
the coordinates and one for producing the domino), 
each play reaches either the safe sink~$\oplus$ 
or the unsafe sink~$\ominus$, which are observable to both players. 
Hence, the game trivially allows for \RCKS.
The construction can be done in time $O(\ell + |\calD|)$.

In summary, we have a linear-time reduction from 
\textsc{Exp-Square Tiling}, 
to the problem of deciding 
whether there exists a joint winning strategy in a game that allows for \RCKS.
According to~Theorem~\ref{thm:corridor-complexity}, this shows that 
our problem is \NEXPTIME-hard.
 
\medskip
For the upper bound and the strategy-construction procedure, 
it would be inconvenient to rely on 
the tracking construction used in the decidability proof of 
Theorem~\ref{thm:synthesis-decidable},
as the number of epistemic structures (over histories of quadratic 
length that are relevant here) 
is already doubly exponential in the size of the game graph.
Instead, 
we introduce an auxiliary representation of the game
which retains the histories at which players attain \CKS and 
is only simply exponential in the size  of the 
input game graph.

\subsection{The abridged game}

For the proof in the reminder of the section, 
let us fix a coordination 
game $\calG = (G, \gamma)$ for $n$ players over a  
game graph~$G$ that allows for \RCKS, with a parity
condition over a set of priorities $C = \{1, \dots, |C| \}$ 
described by the colouring function $\gamma: V \to C$. 
Recall that a play is winning under the parity condition if the least
priority seen infinitely often along a play is even. 
The assumption that the players can observe the
the prioriy coloring is inessential for parity conditions; 
it is only needed for $\omega$-regular conditions that will be
discussed at the end of the section.

The \emph{abridged} game $\widehat{\calG}$ of $\calG$
is a game with perfect information 
for one player 
against Nature. Intuitively, $\widehat{\calG}$ is obtained by
contracting knowledge gaps and recording
only the most significant priority seen 
between two consecutive histories where the players attain 
\CKS.

Concretely, the states of the abridged game graph $\widehat{G}$ 
are pairs $(v, c)$ 
of states $v \in V$ and priorities $c \in C$; for convenience, we also
include a sink~$\ominus$. We shall
refer to the states of $\widehat{G}$
as \emph{positions}, 
to avoid confusion with the ones of $\calG$.
The initial position $(v_0, |C|)$  
corresponds to the initial state of $\calG$
labelled with the least significant priority.
The set of actions consists of all
nonempty subsets $U \subseteq V \times C$ of positions. 
The player has perfect information, 
so the observation function is the identity on $V \times C$.

To define the moves,
we look at the
unravelling $\tck{G}$ up to common knowledge 
of the game graph $G$, as constructed in the proof of 
Proposition~\ref{prop:ck-bounded}. 
Recall that $\tck{G}$ is built from a 
disjoint collection of trees~$(T_v)_{v \in V}$, which are
then connected by identifying 
all leaves with the corresponding roots.
For every state $v \in V$ and 
any joint strategy $t$ over the tree component $T_v$ of
$\tck{G}$, we define the set $\outcome_v( t )$ 
of pairs $(u, d) \in V \times C$, 
for which
there exists a history $\tau$ in~$T_v$ that follows~$t$, 
such that $\tau$ ends at $u$, and the most
significant priority that occurs along $\tau$ is~$d$.
Now, the set of available moves is defined as follows.
For an action $U \subseteq V \times C$ 
there are moves from a position $(v, c)$ 
to every position $(u, d) \in U$ if
there exists a joint strategy $t$ in  $T_v$ 
with $\outcome_v(t) = U$. 
Otherwise, the action leads to the~$\ominus$-sink.
Notice that the moves depend only on the first 
component of the position, that is, on the state and 
not on the priority.

At last, we define a parity condition on $\calG$,
by assigning to every position $(v, c) \in V \times C$ 
the priority~$c$.

The plays of~$G$ and $\widehat{G}$ are related via their summaries. 
Intuitively, this is 
the sequence  
of states reached when
the players attain \CKS in a play,   
together with the most significant priority 
seen along the last knowledge gap.
More precisely, 
for a play $\pi = v_0 \, a_1 v_1\, \dots$ in $G$, 
we look at the subsequence 
of rounds $t_0, t_1, t_2, \dots$
such that, for all $\ell \ge 0$, the players attain 
\CKS at round~$t_\ell$ in $\pi$,
but not at any round~$t$ in between 
$t_{\ell} < t < t_{\ell+1}$.
Next, we associate to each index $\ell > 0$, 
the most significant colour 
that occurred in the gap between $t_\ell$ and $t_{\ell+1}$, setting
$c_{\ell+1} := \min \{ \gamma( v_t ) ~:~ t_{\ell} < t \le t_{\ell+1}\}$.
Now, the \emph{summary} of $\pi$ is the
sequence $[\pi] := v_0, (v_{t_1}, c_1), (v_{t_2}, c_2) \dots$
Notice that for every play $\pi$ in~$G$, 
the summary~$[\pi]$ corresponds to a sequence of states
in~$\widehat{G}$, which is infinite, since  
we assume that~$\pi$ allows for \RCKS.

The notion of summary is defined analogously for histories, 
and it also applies to plays $\hat{\pi}$ in $\widehat{G}$. 
Indeed, $[\hat{\pi}]$ is 
obtained simply by dropping the actions in $\hat{\pi}$.
We say that a play $\pi$ in $G$ \emph{matches} 
a play $\hat{\pi}$ in $\widehat{G}$
if they have the same summary: $[\hat{\pi}] = [\pi]$.

The winning or losing status is preserved among matching
plays. 

\begin{lemma}\label{lem:abridged-play}
If a play $\pi$ of $\calG$ matches a play 
$\hat{\pi}$ of $\hat{\calG}$, then $\pi$ 
is winning if, and only if, $\hat{\pi}$ is winning.
\end{lemma}

\begin{proof}
Let~$c$
be the least priority that appears infinitely often
in~$\pi$. As each knowledge gap in $\pi$ is finite, $c$
appears in infinitely many knowledge gaps in~$\pi$, 
hence it is 
recorded infinitely often in the summary~$[\pi]$. 
Conversely, all priorities that appear
infinitely often in the summary $[\pi]$, also appear infinitely often
in~$\pi$, so $c$ is minimal among them.
In conclusion, the least priority appearing infinitely often 
in the summaries $[\pi] = [\hat{\pi}]$ 
is the same as in the plays $\pi$ and $\hat{\pi}$.  
\end{proof}

\subsection{Reduction to parity games with perfect information}

To use results from the standard literature on parity games, 
it is convenient to view
the abridged game~$\widehat{\calG}$ formally as a 
turn-based game between two players, \emph{Coordinator} and
\emph{Nature}. In contrast to before,
we shall hence regard Nature as 
an actual player with 
proper positions, moves, and strategies.

Towards this, we view the game graph $\widehat{G}$ as a 
bipartite graph, with one partition
$V \times C$ controlled by Coordinator, and a second one formed 
by the nonempty subsets of $V \times C$, controlled by Nature.
The initial position $(v_0, |C|)$ is unchanged.
Coordinator can move from every position
$(v, c) \in V \times C$ 
to a position $U \subseteq V \times C$, 
if $U = \outcome_v(t)$  for some joint strategy $t$ on 
$T_v$, whereas Nature can move 
from every position $U \subseteq V \times C$ 
to any element $(u, d) \in U$. 
The new positions from
$U \subseteq V \times C$ receive the 
least significant priority $|C|$, whereas 
position $(v, c) \in V \times C$ 
have priority~$c$, as before. 

A fundamental result about parity games is that they enjoy 
\emph{positional determinacy}.
A strategy is positional if the choice prescribed 
at a history $\pi$ depends only on the last position in~$\pi$. 
The following theorem was first proved by Emerson and 
Jutla~\cite{EmersonJ91}, a comprehensive 
exposition can be found in the survey of 
Zielonka in~\cite{Zielonka98}.
\begin{theorem}[\cite{EmersonJ91}]
For every parity game with perfect information, one of the two players
has a positional winning strategy.
\end{theorem}

For our setting, positional determinacy means that
in the abridged game~$\widehat{\calG}$, either Coordinator or Nature 
has a winning strategy defined on the set of positions.
This yields witnesses of manageable size for determining 
which player wins the abridged game.

In the following, we argue that positional strategies 
for the abridged game~$\widehat{\calG}$ 
can be translated effectively into strategies on $\calG$, such
that the resulting plays match
in the sense of Lemma~\ref{lem:abridged-play}.

\begin{prop}\label{transfer-abridged}
Let $\calG$ be a game that allows for \RCKS, 
and let $\widehat{G}$ be the abridged game. 
\begin{enumerate}[(i)]
\item For every positional Coordinator strategy $\hat{s}$
in~$\widehat{\calG}$, we can
effectively construct a strategy profile $s$ for the 
grand coalition in~$\calG$ such that, for every
play~$\pi$ that follows~$s$, there exists 
a matching play~$\hat{\pi}$ that follows~$\hat{s}$.
\item For every positional Nature strategy $\hat{r}$
in~$\widehat{\calG}$, 
and every strategy profile~$s$ for the 
coalition in~$\calG$, there exists a play $\pi$ in~$\calG$ 
that follows~$s$ with a matching play~$\hat{\pi}$ 
that follows $\hat{r}$. 
\end{enumerate}
\end{prop}

\begin{proof}

(i) Let $\hat{s} : V \times C \to 2^{V \times C}$ be a positional
strategy for Coordinator in the abridged game $\widehat{\calG}$.
We construct a strategy $s$ for the unravelling $\tck{\calG}$ of 
$\calG$ up to
common knowledge. 
As the two game graphs have the same unravelling, $s$ is also a
strategy for $\cal{G}$.

We can assume that 
the strategy~$\hat{s}$ prescribes for every state~$v \in V$ 
the same choice at all positions $( v, c )$  
independently of the priority.
This is without loss of generality:
Recall that all positions $(v, c)$ in $\widehat{\calG}$ have the same
set of successors~$U$.
If we add a fresh position~$z_v$, of least significant priority, 
from which Coordinator can move to every position in~$U$, 
and replace the outgoing moves from each position 
$(v, c)$ with a move to $z_v$,
the game remains essentially unchanged.  
Whenever Coordinator has a winning strategy for the original game, 
he has one for the modified game. Then, due to positional determinacy,
he also has a positional winning strategy
and its choice at the new position $z_v$ can be transferred 
as a uniform choice to all positions $(v, c)$
in the original game, still yielding a winning strategy.

To transfer the given strategy from~$\widehat{\calG}$ to 
$\tck{\calG}$, we consider for each state $v \in V$ 
the tree component $T_v$ of 
$\tck{\calG}$ separately.
For an arbitrarily chosen colour~$c$, 
we look at the set
$U := \hat{s}( v, c )$ and pick a joint strategy 
$t_v$ on $T_v$ with $\outcome_v( t_v ) = U$.
Now, for every history $\pi$ that ends in $T_v$, we
take the suffix~$\pi_v$ 
contained in $T_v$, that is, we forget the prefix history 
up to entering the tree, and set 
$s( \pi ) = t_v( \pi_v )$. This is a valid strategy profile,
due to the reindexing argument for 
private knowledge undelying Lemma~\ref{lem:reindexing}. 

With~$s$ constructed this way, every play $\pi$ in $G$ 
that follows~$s$
has the same summary  $[\pi] = v_0 \, (v_1, c_1) (v_2, c_2) \dots$
as the play 
$\hat{\pi} = v_0 \, a_1 (v_1, c_1)\, a_2 (v_2, c_2) \dots$ in
$\widehat{G}$  
with actions $a_\ell = \hat{s} (v_\ell, c_\ell)$. 
Hence, $\hat{\pi}$ follows $\hat{s}$ and matches $\pi$,
as required.

\medskip
\noindent
(ii) For the converse, let $\hat{r}: 2^{V \times C} \to V \times C$ 
be a positional strategy for
Nature in $\widehat{G}$ and let $s$ be an arbitrary strategy for
Coordinator in $G$. 
We construct a pair of plays $\pi$ in $G$, and 
$\hat{\pi}$ in $\widehat{\calG}$ with the desired properties.

The construction is by 
induction on the number of knowledge gaps in $\pi$: 
For every $\ell$, we construct a history 
$\pi_\ell$ in $G$ 
with $\ell$ knowledge gaps that follows~$s$ and ends at some
state~$v$, where the players attain \CKS. 
At the same time, 
we construct a matching history $\hat{\pi}_\ell$ 
that follows $\hat{r}$ and ends at 
a position $(v, c)$ in~$\widehat{G}$, associated with the same
state~$v$. 

For the base case, 
both histories $\pi_0$ and $\hat{\pi}_0$ are set to $v_0$. 
For the induction step, suppose that the two histories
$\pi_\ell$ and $\hat{\pi}_\ell$ satisfy the hypothesis, 
and that they end at state $v$ 
and position $(v, c)$, respectively.
We construct a prolongation $\pi_{\ell+1}$ 
that follows $s$ over the $\ell+1$st knowledge gap and 
matches a one-round 
prolongation $\hat{\pi}_{\ell+1}$ of $\hat{\pi}_{\ell}$.
Towards this, we
consider the strategy $t_v$ 
induced by $s$ in the set of histories $\pi_\ell T_v$, that is, the
prolongations of $\pi_\ell$ into the tree component $T_v$ of 
$\tck{G}$. For $U := \outcome_{v}( t_v )$ and
$(u, d) := \hat{r}( U )$, 
there exists a history $\tau$ in $T_v$ that ends at $u$
and has $d$ as most significant priority
after the initial state $v$. 
Now, we update $\pi_{\ell+1} := \pi_{\ell}\tau$ 
and $\hat{\pi}_{\ell+1} := \hat{\pi}_{\ell} U (u, d)$. 
This way, $\pi_{\ell + 1}$ follows $s$ and the players attain \CKS, 
and $\hat{\pi}_{\ell+1}$ follows $\hat{r}$. Moreover,  
the two plays have the same
summary, and $\hat{\pi}_{\ell+1}$ ends at a position corresponding to
the last state of $\pi_{\ell+1}$.

For the infinite plays $\pi$ and $\hat{\pi}$ 
obtained at the limit, we have:
$\pi$ follows $s$ and matches $\hat{\pi}$ which follows $\hat{r}$, as required.

\end{proof}

The 
correspondence between strategies in the
abridged game and in the original game allows us to draw the
following conclusion.

\begin{prop}\label{prop:estimate} 
  Let $\calG$ be a coordination game that
  allows for recurring common knowledge of the state,  
  with $m$ states and a parity winning condition 
  over $d$ priorities. 
  \begin{enumerate}[(i)]
    \item 
      The grand coalition has a joint winning strategy for~$\calG$
      if, and only if, Coordinator has a positional winning strategy 
      for the abridged game~$\widehat{\calG}$, that is a
      perfect-information parity game with
      $md + 2^{md}$ positions and $d$ priorities.  
    \item
      If the grand coalition has a joint winning strategy in $\calG$, 
      then there exists a winning profile of finite-state strategies
      with $2^{\mathcal{O}(m^2\log m)}$ states.
\end{enumerate}
\end{prop}

\begin{proof}
(i) If Coordinator has a positional winning strategy $\hat{s}$ 
in $\widehat{\calG}$, 
then the corresponding profile $s$ according to
Proposition~\ref{transfer-abridged}(i) is
winning in $\calG$, because every play $\pi$ that follows $s$
has a matching play in $\calG$ that follows $\hat{s}$
and is hence winning, which implies that $\pi$ is also
winning, by Lemma~\ref{lem:abridged-play}.

Conversely, assume that there exists a 
joint winning strategy $s$ in $G$.
By Proposition~\ref{transfer-abridged}(ii), 
for any arbitrary positional strategy~$\hat{r}$ of Nature, 
there exists a
play $\hat{\pi}$ that follows~$\hat{r}$ and matches some play~$\pi$ 
in $G$ which follows $s$ and thus wins. Hence, $\hat{\pi}$ is also
winning for Coordinator, by
Lemma~\ref{lem:abridged-play} which means that 
$\hat{r}$ in not winning for Nature.
By positional determinacy, it follows that
Coordinator has a positional winning strategy in $\widehat{G}$.

The state space of the abridged game  
is $V\times C \cup 2^{V \times C}$, it has $md + 2^{md}$ positions; 
the number $d$ of priorities is 
as in $\calG$.

\medskip\noindent
(ii) Let $\hat{s}: V \times C \to 2^{V \times C}$ 
be a winning strategy for Coordinator in the abridged game
$\widehat{\calG}$. 
As in the proof of Proposition~\ref{transfer-abridged}(i),
we assume, without loss of generality, 
that the strategy prescribes 
the same move $U = \hat{s}( v, c )$, at all 
positions corresponding to $v$, independently of the priority $c$; 
for each of the $m$ states $v \in V$,
the move $\hat{s}( v, c )$ is translated into
an imperfect-information strategy $t_v$ 
on the tree component $T_v$ of $\tck{G}$. 
We use these local strategies to
construct a joint winning strategy $s$
for the grand coalition in~$G$ as follows.

For each player~$i$, the component strategy $s^i$ is implemented 
by a reactive procedure that maintains, along the infinite sequence of
input observations, a 
record $(v, \rho^i)$  of the last state~$v$ 
about which the players attained common
knowledge and the observation history $\rho^i$ along the 
subsequent knowledge gap. 
Initially $v$ is set to $v_0$ and $\rho^i$ to $\beta^i( v_0 )$.
In each step, 
the procedure returns the action $a^i := t_v^i( \rho )$, inputs the next
observation $b^i$, and repeats with $\rho^i a^i b^i$ as a new
value for $\rho^i$, unless
this corresponds to a history in $G_v$ 
at which the players attain common knowledge of the current state~$v'$.
In that case, the root $v$ is replaced with $v'$ 
and the new value of $\rho^i$ becomes $b^i = \beta^i( v')$.
Each local strategy~$t_v^i$ can be represented
by a (tree shaped) automaton that outputs actions in response to observation 
sequences along knowledge gaps\,---\,of
length at most~$m^2$, by Theorem~\ref{thm:ck-bounded}. 
Since there are no more observations than game states, 
$m^{m^2}$ automaton states are
sufficent to store these responses, as well as the set of histories at
which the players attain \CKS.   
Globally, the strategy~$s^i$ of each player~$i$ combines $m$ local
strategies~$t_v^i$. 
Hence, we need at most $m \cdot m^{m^2} = 2^{\mathcal{O}(m^2 \log m)}$ 
many states to represent each component of the profile~$s$ 
by a strategy automaton.
\end{proof}

\subsection{Complexity}

A nondeterministic procedure for deciding whether there exists a joint
winning strategy in a game $\calG$ 
with \RCKS, according to Proposition~\ref{prop:estimate},  
can guess the abridged game $\widehat{G}$ and determine whether
Coordinator has a winning strategy in the obtained parity game with
perfect information.
The complexity is dominated by the verification of the transition
relations between Coordinator positions $(v, c)$ and Nature positions
$U \subseteq V \times C$, which involves guessing a witnessing
strategy profile $t_v$ over the tree $T_v$ 
such that $\outcome_v( t_v ) = U$. 
As we pointed out in the proof of Proposition~\ref{prop:estimate}, 
for a game $\calG$ of size $m$, such a
strategy~$t_v$ can be 
represented by a collection of $n$ trees of size 
$2^{\mathcal{O}(m^2\log m)}$, one for every player. Once the local
strategy trees $t^i$ are guessed, the verification that 
$\outcome_v( t ) = U$ is done in time linear in their size. 
Given the abridged game, a winning strategy for Coordinator can 
be guessed and verified in nondeterministic linear time 
with respect to the size 
$md + 2^{md}$ of~$\widehat{\calG}$ where $d$ is the number of priorites.
Overall, the procedure runs in $\NTIME(2^{\mathcal{O}(m^2\log m)})$,
that is, nondeterministic exponential time.

With a deterministic procedure, 
the abridged game can be constructed by exhaustive search over
witnessing strategies over the component trees in $\tck{G}$
in time $2^{2^{\mathcal{O}(m^2\log m)}}$. Once this is done, 
winning strategies for the obtained parity game
$\widehat{\calG}$ can be constructed in time
$\mathcal{O}(2^{md^2})$ using the basic iterative 
algorithm presented by Zielonka in~\cite{Zielonka98}. 
This concludes the proof of 
Theorem~\ref{thm:ck-complexity}.

\subsection{Observable $\omega$-regular conditions}

In view of applying our results to the practice of 
automated verification and design, we
briefly outline a procedure for synthesing distributed winning
strategies in games with winning conditions
expressed by standard specification formalisms rather than by
parity conditions. 

It is well known that 
every $\omega$-regular language of infinite words can be recognised
by a deterministic automaton with parity acceptance
condition~\cite{Thomas90}.
Given a game $\calG = (G, W)$ with an $\omega$-regular
winning condition $W$ represented by a deterministic
automaton~$\calA$ over attributes (colours) of the game states, 
we construct a parity game $\calG'$ 
on the game graph obtained as the 
synchronised product of $G$ with~$\calA$ 
and define a priority colouring that
associates to every state of the product graph 
$G'$ the priority of the automaton state in
its second component. 
Informally, this corresponds to running the automaton along the plays
in $G$ to monitor the $\omega$-regular winning condition over 
the (colouring of) game states by a parity condition 
over the automaton states. 
Then, the synthesis problem for
the original game $\calG$ reduces 
to the synthesis problem for the game $\calG'$ with a parity
condition. This transformation works as in the case of
perfect-information games detailed in~\cite[Chapter 2]{GraedelThoWil02}.

Now, let us assume that the game $G$ at the outset allows for \RCKS 
and that the winning condition $W$ is expressed
by an observable colouring. 
The latter assumption implies that, at every history in the product game $G'$, 
the current state of the automaton monitoring the winning condition is
common knowledge among the players. 
Since along every play in~$G$ the players attain \CKS infinitely often, 
and there are only finitely many automaton states, 
it follows that the players also attain common knowledge of the 
product state in $G'$ infinitely often.
Hence, the product game graph $G'$ allows
for \RCKS.
In conclusion, the synthesis problem for games that allow for \RCKS
and with observable $\omega$-regular conditions represented by deterministic
automata can be solved with the same generic complexity as games with
parity conditions. 

\begin{corol}\label{cor:reg-complexity}
For games that allow for recurring common knowledge of the
state, with observable~$\omega$-regular winning conditions 
represented by a deterministic parity automata,
\begin{enumerate}[(i)]
\item \label{itm:reg-decision-complexity}
the problem of deciding whether there exists a joint winning strategy 
is $\NEXPTIME$-complete;   
\item \label{itm:reg-construction-complexity}
if joint winning strategies exist, 
there also exists a winning profile of
finite-state strategies of at most exponential size, which 
can be synthesised in $2\textrm{-}\EXPTIME$. 
\end{enumerate}
\end{corol}

To obtain precise upper bounds, we need to take
into account that the product construction
 increases the size of the game graph by a factor
corresponding to the size of the deterministic automaton.
More generally, for 
winning conditions specified in common verification formalisms, 
e.g., $\omega$-regular expressions, PDL, LTL, or nondeterministic automata,
we can apply the standard techniques for transforming the
specification into deterministic parity automaton 
to establish the complexity of the
synthesis problem for games with \RCKS.

\section{Conclusion}

We identified a new class of games with imperfect information 
for which the distributed synthesis problem 
can be solved effectively: 
It is decidable whether distributed winning strategies exist and, 
if so, a profile of finite-state winning strategies can be computed.  
Our procedure for solving the distributed synthesis problem
for infinite games with~\RCKS under 
parity winning condition matches the lower complexity bounds
for solving the particular case of  
multi-player safety games of finite horizon.
Whether a game belongs to the class is decidable efficiently.

Known decidable classes from the distributed-systems literature
rely on decomposing the global synthesis problem into  
separate instances, each involving only one player and the environment 
(Nature), that can be solved by automata-theoretic techniques 
for zero-sum games. 
The approach proved successful in several cases
where the dependencies between the behaviour of players are restricted,
typically by a hierarchical information order among them. 
Prominent examples are weakly-ordered architectures~\cite{KupfermanVar01} 
and doubly-flanked pipelines~\cite{MadhusudanT01}, both subsumed by 
Coordination Logic~\cite{FinkbeinerSch10}, or, 
in the asynchronous setting, 
well-connected architectures \cite{GastinSZ09}. 
As the synthesis procedures 
rely on solving nested instances for all players, 
these classes generally display nonelementary complexity. 
A class of more moderate \NEXPTIME complexity was recently proposed by 
Chatterjee et al.~\cite{CHOP13}. 
Here, the winning conditions are restricted to 
ensure an even higher degree of
independence: Essentially, 
each player can achieve her part of the global objective independently 
of the others. 

Our approach is orthogonal to the idea of
decomposing games into two-player zero-sum instances.
Instead of restricting the order of information 
or the game objective, our
decidability condition requires that players attain common knowledge 
of the game state infinitely often along every play. 
Intuitively, this allows to decompose the game tree 
into a sequence of time slices (the gaps) 
that can be solved independently, 
rather than reducing it to parallel zero-sum instances 
for each individual players. 
As a most simple case, our class subsumes 
repeated safety games 
of finite horizon with imperfect information, 
where the initial state (re-entered at each repetition) is observable.
Since safety games of finite horizon are
already \NEXPTIME hard to solve, this justifies the lower bound for our
solution procedure. Nevertheless, it is rewarding to see that
the synthesis problem for arbitrary games with \RCKS 
has a matching \NEXPTIME upper bound, 
in spite of covering much more general examples of games.

For instance, the class 
captures the interaction scenarios that proceed in
phases where imperfect information can arise and evolve in any form, provided
that at the outcome of each phase the participants synchronise 
in some state that will become common knowledge among them. 
This may be guaranteed explicitely, by restrictin to phase
games on acyclic graphs with observable exit states, or
implicitely, by ensuring that the 
players attain \CKS due to the structure of game graph.
We believe that designs of distributed systems 
tend to follow such patterns naturally, as developers 
introduce breakpoints or synchronisation barriers for monitoring
their system. One challenge is to develop concrete
communication schemes for distributed systems that yield 
games with~\RCKS. As an application to decentralised control, 
it will be interesting to identify
conditions under which common knowledge of an event can be
approximated by (a finite degree of) mutual knowledge.  

Apart from direct applications, we hope that our contribution may
help to demystify the subject of imperfect information in multiplayer
games haunted by the discouraging complexity results for the
general case. As pointed out by Muscholl and Walukiewicz in their
concise survey on  
distributed synthesis~\cite{MuschollWal14}, currently
we have no evidence that the constructions causing
undecidability in general game models 
would arise in real-world systems, but
we have no systematic justification for ruling them out either.
We may paraphrase the insights of the present
article by concluding
that, if imperfect information is admitted only as a temporary perturbation of
perfect information, then we can rule out undecidable situations.
It remains to investigate whether this intuition applies to
further relevant forms of 
perturbation in the information structure of games as, for
instance, communication delays or incomplete but perfect information.

\noindent\paragraph{Acknowledgements} 
This work was supported by the 
European Union Seventh Framework Programme under Grant
Agreement 601148 (CASSTING) and by the Indo-French Formal Methods
Lab (LIA Informel).
A preliminary report~\cite{BerwangerMat14} was presented at the
2014 Workshop on Strategic Reasoning.

\bibliographystyle{elsarticle-num}

\bibliography{all}

\end{document}